\documentclass[letterpaper]{article} 
\usepackage{aaai2026}  
\usepackage{times}  
\usepackage{helvet}  
\usepackage{courier}  
\usepackage[hyphens]{url}  
\usepackage{graphicx} 
\usepackage{natbib}  
\usepackage{caption} 
\title{A Graph-Theoretical Perspective on Law Design for Multiagent Systems}

\usepackage{bibentry}

\usepackage{preamble}

\nocopyright     

\author {
    Qi Shi,
    Pavel Naumov
}
\affiliations {
    University of Southampton, United Kingdom\\
    qi.shi@soton.ac.uk, p.naumov@soton.ac.uk 
}

\begin{document}

\maketitle

\begin{abstract}
A law in a multiagent system is a set of constraints imposed on agents' behaviours to avoid undesirable outcomes. The paper considers two types of laws: useful laws that, if followed, completely eliminate the undesirable outcomes and gap-free laws that guarantee that at least one agent can be held responsible each time an undesirable outcome occurs. 
In both cases, we study the problem of finding a law that achieves the desired result by imposing the minimum restrictions.

We prove that, for both types of laws, the minimisation problem is NP-hard even in the simple case of one-shot concurrent interactions. We also show that the approximation algorithm for the vertex cover problem in hypergraphs could be used to efficiently approximate the minimum laws in both cases.

\end{abstract}



\section{Introduction}\label{sec:introduction}

Suppose that three factories $a$, $b$, and $c$ need to dump the same type of pollutant into a river. Each factory must dump once every three days. The assimilative capacity of the river endures at most two factories dumping per day. Otherwise, the fish in the river would be killed.
To avoid the death of the fish, the local government would like to set a law that regulates the dumping activity.
Building on the well-known legal maxim ``everything which is not forbidden is allowed'', we assume that \textit{a law serves solely to identify the actions from which agents must abstain}.
This is in line with the liberal rule-of-law perspective: a key virtue of the rule of law is the protection of individual freedom \cite{hayek1944road,raz1979rule}.

In a simple form, a dumping law can assign each factory a fixed dumping day within a recurring three-day cycle by banning it from dumping on the other two days.
For instance, the law specifies the set
$
    L_0=\{d_a^1,d_a^2,d_b^2,d_b^3,d_c^1,d_c^3\}
$
of \textit{banned} dumping actions, where $d_x^i$ represents the action that factory $x$ dumping on the $i^\text{th}$ day in each three-day cycle.
Under this law, factory $a$ dumps on the third day ($d_a^3$), factory $b$ dumps on the first day ($d_b^1$), and factory $c$ dumps on the second day ($d_c^2$) of each three-day cycle.
In other words, only one factory dumps on each day, and it is guaranteed that the fish will not be killed.
In this case, we say that law $L_0$ is \textit{useful} in terms of prohibiting the death of the fish.
In general, we say that a law is {\bf\em useful} if \textit{the prohibited outcomes shall never appear when every agent obeys the law}. 
The term ``useful'' is adopted from \cite{shoham1995social}, a pioneering work on law design in multiagent systems.

It is easily observable that the law $L_0$ above unnecessarily constrains the dumping behaviour of the factories.
Note that, to avoid the death of the fish, it suffices to ensure that not all three factories dump on the same day.
In other words, for any given day, the law only needs to prevent one factory from dumping.
This means that the set $L_1=\{d_a^1,d_b^2,d_c^3\}$ is also a useful law.
However, law $L_1$ allows each factory one more day to dump in each three-day cycle, providing more flexibility toward the factories' production activities.

Observe that laws $L_0$ and $L_1$ are both useful, while $L_1$ sets fewer constraints than $L_0$ (\textit{i.e.} $L_1\subsetneq L_0$). In this case, we say that $L_1$ is a {\bf\em useful reduction} of law $L_0$.
Also, notice that any further reduction of law $L_1$ is no longer useful. 
For example, the reduced law $L_2=\{d_a^1,d_b^2\}$ of law $L_1$ allows all factories to dump on the third day and kill the fish.
In other words, law $L_1$ satisfies a minimality property regarding usefulness.
In general, we say that a law is {\bf\em minimal-useful} if \textit{it is useful but cannot be further reduced while keeping usefulness}.
Considering the minimality of law captures the idea of setting minimal constraints on society, which is in line with the opinion that ``the minimal state is the most extensive state that can be justified. Any state more extensive violates people's rights ...'' \cite{nozick1974anarchy}.

Essentially, the usefulness of a law captures the \textit{ability of the law to prevent} the undesirable outcomes.
This is a typical focus in the literature on law design in multiagent systems.
While reliable, the usefulness requirement excludes the possibility of \textit{coordination} among agents as a means of prevention.
In our example, even without a law, the three factories can still negotiate a dumping plan to avoid the death of the fish.
On one hand, a negotiated plan could be more adaptive to the production of the factories than a law and thus bring higher efficiency.
On the other hand, a law should be stable \cite{jefferson1787unstability,raz1979rule}. 
In contrast, coordination achieved by negotiation is flexible and thus can better accommodate changes in production.
Although appealing, coordination is not always realisable in multiagent settings. 
In our case, the three factories may not be able to reach an agreement on the dumping plan due to the lack of communication or the conflict of interest.
In general, coordination is even harder to achieve in complex multiagent systems, particularly when different types of agents exist (\textit{e.g.} a traffic system including autonomous vehicles and human drivers, or a ranch including sheep, sheepdogs, herders and wolves).

In this paper, we relax the usefulness requirement and consider laws that can simultaneously \textit{accommodate coordination and the failure of coordination}. 
Informally, under such laws, prohibited outcomes may appear when the agents \textit{all obey the law but do not coordinate}. 
Nevertheless, there is always \textit{at least one {\bf\em principal agent} who has a {\bf\em safe action} that is lawful, and enables the prevention of prohibited outcomes without the need for coordination}.
As an example, law $L_2$ satisfies the above property.
Recall that law $L_2$ only bans factory $a$ from dumping on the first day and factory $b$ from dumping on the second day.
Consequently, even if all factories comply with law $L_2$, without coordination, the death of the fish may still happen when they simultaneously dump on the third day.
However, factory $c$ has a safe action under law $L_2$ (\textit{i.e.} dumping on either the first or the second day) that can solely prevent the death of the fish as long as factories $a$ and $b$ obey law $L_2$.
By this means, law $L_2$ leaves factory $c$ a chance to dump on the third day while keeping the fish alive through negotiation with other factories.
Meanwhile, if the negotiation fails, factory $c$ can still prevent the death of the fish on its own.
In other words, factory $c$ is a principal agent under law $L_2$.

Observe that, under such laws, if a prohibited outcome finally appears, then either there is an agent whose action breaks the law, or the principal agent fails to prevent it.
In the former case, we say that the agent who breaks the law bears {\bf\em legal responsibility}.
In the latter case, we say that the principal agent who could have prevented the prohibited outcomes with a safe action but fails to do so bears {\bf\em counterfactual responsibility}.
In our example, if the fish is killed because all three factories dump on the first day, then factory $a$ breaks law $L_2$ and thus $a$ is legally responsible; if it happens on the second day, then factory $b$ breaks law $L_2$ and thus $b$ is legally responsible; if it happens on the third day, then factory $c$, the principal agent under law $L_2$, fails to utilise her safe action to prevent it, and thus $c$ is counterfactually responsible.
In a word, under those laws, \textit{if a prohibited outcome happens, then there is at least one agent either legally or counterfactually responsible}.
That is, a responsible agent can always be identified for any prohibited outcome.
In this sense, we refer to such laws {\bf\em responsibility gap-free} (\textit{abbr.} {\bf\em gap-free}).

To further clarify our terminology, \textit{counterfactual responsibility} is usually regarded as a form of moral responsibility \cite{robb2023moral}. It captures the {\em principle of alternative possibilities}~\cite{frankfurt1969alternate}: \textit{a person is morally responsible for what she has done only if she could have done otherwise}.
In the recent literature on responsibility in multiagent systems, the component ``could have done otherwise'' is commonly interpreted as an agent’s strategic ability to guarantee prevention irrespective of the behaviour of other agents \cite{naumov2019blameworthiness,yazdanpanah2019strategic,baier2021game,shi2024responsibility}.
Note that the former discussion about the principal agent’s safe action to prevent prohibited outcomes aligns with such an interpretation.
In this sense, our use of the term \textit{counterfactual responsibility} in this paper (\textit{i.e.} principal agents who fail to prevent are counterfactually responsible) is consistent with its treatment in the existing literature.
As for \textit{responsibility gap}, also called \textit{responsibility void}, it is one of the important topics discussed in the ethics literature, especially in the context of artificial intelligence \cite{matthias2004responsibility,braham2011responsibility,duijf2018responsibility,braham2018voids,burton2020mind,gunkel2020mind,langer2021we,goetze2022mind}.
Informally, responsibility gap captures the situation where an undesired outcome occurs but no agent can be held responsible, which is usually regarded as ``unwanted'' \cite{hiller2022axiomatic}.
Correspondingly, the term \textit{gap-freeness} in this paper precisely captures the \textit{absence of responsibility gaps} as discussed in the literature and serves as a desired property of multiagent systems.

Back to the point, by the former analysis, law $L_2$ is gap-free. It sets fewer constraints on the agents' actions than the minimal-useful law $L_1$ and allows for potential coordination.
However, law $L_2$ is not minimal yet in terms of gap-freeness.
Let us consider the law $L_3=\{d_a^1\}$, a reduction of law $L_2$.
Since law $L_3$ bans factory $a$ from dumping on the first day, as long as factory $b$ dumps on that day (\textit{i.e.} safe action), the fish would survive regardless of factory $c$'s choice. The same applies to factory $c$.
In other words, both factories $b$ and $c$ are principal agents under law $L_3$.
In addition, under law $L_3$, if the fish is killed because all factories dump on the first day, then factory $a$ is legally responsible; if it happens on the second or the third day, then factories $b$ and $c$ are both counterfactually responsible.
This implies the gap-freeness of law $L_3$. On top of it, we say that law $L_3$ is a {\bf\em gap-free reduction} of law $L_2$ because $L_3\subsetneq L_2$. 
Moreover, observe that the law $L_4=\varnothing$ is not gap-free because the death of the fish may happen, but no agent has a safe action that can individually prevent it under law $L_4$.
This shows that any further reduction of law $L_3$ is not gap-free, and thus law $L_3$ satisfies a minimality property in terms of gap-freeness.
In general, we say that a law is {\bf\em minimal-gap-free} if \textit{it is gap-free but cannot be further reduced while keeping gap-freeness}.

\paragraph{Contribution}
In this paper, we investigate the \textit{design of useful laws and gap-free laws in multiagent systems}.
In particular, we model multiagent systems as one-shot concurrent games and interpret the law design problem in a graph-theoretical perspective.
Specifically, we first formalise the concepts of usefulness and gap-freeness in Section~\ref{sec:formalisation}.
Then, in Section~\ref{sec:useful law}, we establish the equivalence between the useful law design problem and the vertex cover problem in hypergraphs (\textit{a.k.a.} hitting set problem) by providing polynomial-time reductions between them.
After that, in Section~\ref{sec:gap-free law}, we show that the gap-free law design problem is at least as hard as the useful law design problem, and further present a method for solving the gap-free law design problem by reducing it to the vertex cover problem in hypergraphs.
Note that the vertex cover problem is one of the most important problems in complexity theory, listed as one of the 21 NP-complete problems by \citet{karp1972reducibility}.
Extensive research has focused on its hardness and approximation \cite{chvatal1979greedy,bar1985local,slavik1996tight,feige1998threshold}.

\paragraph{Novelty}
By reducing the law design problems to the vertex cover problem, we make it possible to \textit{tackle the computational intractability in the law design problems using approximation techniques}.
This is distinct from the literature on law design (\textit{i.e.} norm synthesis) in multiagent systems.
In fact, there has been a rich amount of studies about \textit{normative system}, where laws/norms are used to regulate the multiagent systems.
They mainly use first-order logic \cite{shoham1995social,fitoussi2000choosing} or modal logic \cite{van2007social,aagotnes2009temporal} to describe the concerned properties of a system and use deontic logic to capture laws/norms \cite{alechina2018norm}.
In this way, they are capable of modelling more complex systems than one-shot concurrent games.
Meanwhile, they study both offline design (\textit{i.e.} design-time norm synthesis) \cite{fitoussi2000choosing,aagotnes2012conservative} and online design (\textit{i.e.} run-time norm synthesis) \cite{morales2011using,morales2013automated,morales2014minimality,riad2022run}, considering both static and dynamic norms that may evolve \cite{alechina2022automatic,riad2022run}.
Depending on the difference in object and formalisation, the computational complexity of their problems ranged from NP-complete \cite{shoham1995social} to beyond EXPTIME \cite{perelli2019enforcing,galimullin2024synthesizing}. 
Although few of them try to address the complexity by heuristics \cite{christelis2009automated} or by degenerating their problems into optimisation problems \cite{aagotnes2010optimal,wu2022bayesian}, \textit{none have tackled the intractability by considering inapproximability in addition to approximation, as done in this paper}. 
As a by-product, our attempt at gap-free law design offers an applicable way to \textit{address the responsibility-gap concern} in the literature.

\section{Formalisation}\label{sec:formalisation}

In this section, we first define a (one-shot current) game to model multiagent systems. 
Then, we formalise the concepts of law, usefulness, and gap-freeness as discussed in the previous section.
After that, we formally define hypergraphs that are used later to solve the law design problems.

In the rest of this paper, by $|S|$ we denote the size of a set $S$; by $\prod F$ and $\bigcup F$ we denote, respectively, the Cartesian product and the union of all sets in a family $F$.
For an indexed set $\alpha=\{\alpha_i\}_{i\in I}$, by $\mathcal{S}(\alpha)=\{\alpha_i\mid i\in I\}$ we denote the support set of $\alpha$ that forgets the order and multiplicity.

\begin{definition}\label{df:game}
A \textbf{game} is a tuple $(\mathcal{A},\Delta,\mathbb{P})$ such that
\begin{enumerate}
\item $\mathcal{A}$ is a nonempty finite set of agents;\label{dfitem:game agents}
\item $\Delta=\{\Delta_a\}_{a\in\mathcal{A}}$ is a family of sets of \textbf{actions}, where $\Delta_a$ is a finite set of actions available to agent $a$;\label{dfitem:game actions}
\item $\mathbb{P}\subseteq\prod\Delta$ is the \textbf{prohibition}.\label{dfitem:game prohibitions}
\end{enumerate}
\end{definition}

In a game $(\mathcal{A},\Delta,\mathbb{P})$, each profile $\delta\in\prod\Delta$ represents an \textit{outcome}.
Instead of utilities of agents\footnote{
We do not yet consider the incentives of agents in the law design problem, and thus get rid of the utility functions in our model.
}, we consider a set $\mathbb{P}$ of prohibited outcomes.
The purpose of designing laws is to avoid the game ending in the prohibited outcomes.
For instance, in the factory example, all factories dumping on the first day represents an outcome where the fish is killed, and the set of prohibited outcomes consists of the three cases where all agents dump on the same day during each three-day cycle.
A law is established to avoid the death of the fish, that is, to prevent any of the prohibited outcomes.

Technically, item~\ref{dfitem:game actions} of Definition~\ref{df:game} allows the action set $\Delta_a$ to be empty for any agent $a\in\mathcal{A}$.
This is made for mathematical convenience.
Moreover, a profile $\delta\in\prod\Delta$ is essentially a set indexed by every agent $a\in\mathcal{A}$ (\textit{i.e.} $\delta=\{\delta_a\}_{a\in\mathcal{A}}$).
Correspondingly,
\begin{equation}\label{eq:profile support set}
    \mathcal{S}(\delta)=\{\delta_a\mid a\in\mathcal{A}\}
\end{equation}
is the set of actions taken by the agents under a profile $\delta$.
Also, the action sets of different agents are not necessarily identical or disjoint.
This allows greater flexibility in modelling multiagent systems that simultaneously include multiple agents of the same type who are likely to share the same action space, and agents of different types whose available actions usually differ.
More importantly, the laws that set bans on actions, as discussed in Section~\ref{sec:introduction} and formally defined below, will be \textit{agent-independent}.
In other words, it is impossible to forbid one agent from using an action while permitting another agent to use the same action.
In this sense, \textit{fairness} is embedded in our formalisation.

\begin{definition}\label{df:law in game}
A law in a game $(\mathcal{A},\Delta,\mathbb{P})$ is an arbitrary set $L\subseteq\bigcup\Delta$ of actions.
\end{definition}

For two laws $L,L'$ in the same game, we say that $L$ is a {\bf\em reduction} of $L'$ if $L\subseteq L'$.
Note that, as stated in Section~\ref{sec:introduction}, agents can break a law.
However, in a hypothetical situation where all agents obey the law, the agents indeed play a ``subgame'' where only lawful actions are available.
We call such a ``subgame'' \textit{law-imposed game} and formalise it below.

\begin{definition}\label{df:law-imposed game}
For a law $L$ in a game $(\mathcal{A},\Delta,\mathbb{P})$, the \textbf{law-imposed} game is the game $(\mathcal{A},\Delta^L,\mathbb{P}^L)$ where
\begin{enumerate}
\item $\Delta^L=\{\Delta_a^L\}_{a\in\mathcal{A}}$ such that $\Delta_a^L=\Delta_a\setminus L$ is the set of \textbf{lawful action}s of agent $a$;\label{dfitem:law-imposed game 1}
\item $\mathbb{P}^L=\mathbb{P}\cap\prod\Delta^L$ is the \textbf{law-imposed prohibition}.\label{dfitem:law-imposed game 2}
\end{enumerate}
\end{definition}

Informally, for a given game and a given law, in the law-imposed game, the actions are the \textit{lawful actions}, the profiles are the {\em lawful profiles}, and the prohibition consists of the lawful profiles prohibited in the original game.

\subsection{Usefulness of Law}\label{sec:formalisation-useful}

Let us now formalise usefulness. 
Recall that, as discussed in Section~\ref{sec:introduction}, a law is useful if the prohibited outcomes never appear when every agent obeys the law.
That means none of the lawful profiles are prohibited, which is formalised below.

\begin{definition}\label{df:useful law}
A law $L$ in a game $(\mathcal{A},\Delta,\mathbb{P})$ is called \textbf{useful} if $\mathbb{P}^L=\varnothing$ in the law-imposed game.
\end{definition}

The next lemma characterises when a law is useful: every prohibited profile consists of at least one banned action by the law.
See its formal proof in Appendix~\ref{app_sec:proof of Lemma useful law}.

\begin{lemma}\label{lm:useful law}
A law $L$ in a game $(\mathcal{A},\Delta,\mathbb{P})$ is useful if and only if $L\cap\mathcal{S}(\delta)\neq\varnothing$ for each profile $\delta\in\mathbb{P}$.
\end{lemma}

Next, we consider the minimality of a useful law. Our discussion about it in Section~\ref{sec:introduction} can be formally captured below.

\begin{definition}\label{df:minimal useful law}
A law $L$ in a game is \textbf{minimal-useful} if $L$ is useful and no law $L'\subsetneq L$ is useful in the same game.
\end{definition}

Motivated by the pursuit of minimum constraints on society, if a law is useful but not minimal, then we would reduce it to leave as much freedom (\textit{i.e.} lawful actions) as possible. 
This is formally captured as the minimality of useful reduction in the definition below.

\begin{definition}\label{df:minimum useful reduction}
For a useful law $L$ in a game, 
\begin{enumerate}
    \item a \textbf{useful reduction} of $L$ is a useful law $L'\subseteq L$;\label{dfitem:useful reduction}
    \item a useful reduction $L'$ of $L$ is called \textbf{minimum} if there is no useful reduction $L''$ of $L$ such that $|L''|< |L'|$.\label{dfitem:minimum useful reduction}
\end{enumerate}
\end{definition}

Note that the approach of reducing an existing law, rather than crafting a new one from scratch, aligns with the view that laws should be stable \cite{jefferson1787unstability,raz1979rule}, as it ensures that \textit{actions previously permitted remain permitted}.
This is indeed a type of \textit{Pareto optimisation}.
Also, it is not hard to deduce from Lemma~\ref{lm:useful law} that any further reduction of a non-useful law is also non-useful, which is why we consider only the reduction of useful laws here.
Moreover, observe that the law which bans all actions (\textit{i.e.} $L=\bigcup\Delta$) is useful by Lemma~\ref{lm:useful law}, and thus crafting a new useful law from scratch is equivalent to getting a useful reduction of the law $L=\bigcup\Delta$. 
In this sense, \textit{minimising an existing useful law encompasses the task of minimum useful law design}.

Notably, the law $L=\bigcup\Delta$ may not be reducible while preserving usefulness.
For example, consider a matching-pennies game where two agents have the same action space $\{\mathit{head},\mathit{tail}\}$ and the outcomes $(\mathit{head},\mathit{head})$ and $(\mathit{tail},\mathit{tail})$ are prohibited.
In this case, the set $\{\mathit{head},\mathit{tail}\}$ is the only useful law by Lemma~\ref{lm:useful law} and thus minimal-useful.
In other words, some games may admit no useful law that permits any lawful actions.
However, such situations are uncommon in real-world scenarios, where every agent usually has access to a \textit{default action} that can avoid undesirable outcomes.
As a result, a law that permits only default actions will be useful, and any useful reduction of such a law continues to allow those default actions.
In the worst-case scenario where no suitable default action exists, a useful initial law can still be constructed by distinguishing between different agents' actions (\textit{i.e.} by making their action sets disjoint), albeit at the cost of symmetry.
Then, a law that restricts each agent to a single action leading to a fixed non-prohibited outcome (\textit{e.g.} banning one agent from choosing $\mathit{head}$ and the other from choosing $\mathit{tail}$ in the matching-pennies game) is always useful.
This approach mirrors how traffic lights prevent collisions at intersections: when east-west traffic is permitted to proceed, north-south traffic is banned, and vice versa.

\subsection{Gap-Freeness of Law}\label{sec:formalisation-gap free}

In this subsection, we formalise gap-freeness.
As discussed in Section~\ref{sec:introduction}, a key feature of gap-freeness is the presence of a principal agent who can individually prevent the prohibited outcomes.
In our one-shot game setting, such an \textit{ability to prevent} is achieved with a ``safe'' action that guarantees non-prohibited outcomes.
We formally capture this as follows.

\begin{definition}\label{df:safe action}
For a game $(\mathcal{A},\Delta,\mathbb{P})$ and an agent $a\in\mathcal{A}$, an action $d\in\bigcup\Delta$ is called a \textbf{safe action of agent $a$} if $d\in\Delta_a$ and 
$\delta_a\neq d$ for each profile $\delta\in\mathbb{P}$.
\end{definition}

Informally, action $d$ is available to agent $a$, but she does not choose $d$ in any prohibited outcome.
In other words, every potential outcome when agent $a$ chooses action $d$ is not prohibited.
As a result, agent $a$ can prevent the prohibited outcomes by choosing action $d$.
Also note that a safe action $d$ of agent $a$ is not necessarily safe for another agent $b$, even if $d\in\Delta_a\cap\Delta_b$.
This arises from the asymmetry across agents in the prohibition as specified in item~\ref{dfitem:game prohibitions} of Definition~\ref{df:game}.

However, as noted in Section~\ref{sec:introduction}, with a law in place, a principal agent needs only a lawful action that prevents the prohibited outcomes when others act lawfully. 
That is essentially a \textit{safe action in the law-imposed game}.
Next, we capture the legal responsibility of an agent breaking a law and the counterfactual responsibility of a principal agent failing to prevent.
In particular, we say that an agent is responsible if a prohibited outcome happens and she is responsible either legally or counterfactually, as defined below. 

\begin{definition}\label{df:responsibility}
In a prohibited profile $\delta\in \mathbb{P}$ of a game $(\mathcal{A},\Delta,\mathbb{P})$, an agent $a\in\mathcal{A}$ is \textbf{responsible} under law $L$ if
\begin{enumerate}
    \item (legally) $\delta_a\in L$, \textbf{or} \label{dfitem:legal responsibility}
    \item (counterfactually) $\mathcal{S}(\delta)\cap L=\varnothing$ and a safe action of agent $a$ exists in the law-imposed game $(\mathcal{A},\Delta^L,\mathbb{P}^L)$. \label{dfitem:counterfactual responsibility}
\end{enumerate}
\end{definition}

Informally, in a prohibited outcome $\delta$, agent $a$ is legally responsible if her action breaks the law (\textit{i.e.} $\delta_a\in L$);
if no agent breaks the law (\textit{i.e.} $\mathcal{S}(\delta)\cap L=\varnothing$) but agent $a$ is a principal agent under the law (\textit{i.e.} has a safe action in the law-imposed game), then $a$ is counterfactually responsible.

The next lemma characterises when an action is a safe action of an agent under a law: the law should make the action lawful and make each prohibited outcome where the agent takes the action unlawful.
See Appendix~\ref{app_sec:proof of Lemma safe action in law imposed game} for its proof.

\begin{lemma}\label{lm:safe action in law imposed game}
An action $d\in\bigcup\Delta$ is a safe action of an agent $a\in\mathcal{A}$ in the law-imposed game $(\mathcal{A},\Delta^L,\mathbb{P}^L)$ if and only if $d\in\Delta_a^L$ and $L\cap\mathcal{S}(\delta)\neq\varnothing$ for each $\delta\in\mathbb{P}$ such that $\delta_a=d$.
\end{lemma}

The next definition formalises the gap-freeness property.

\begin{definition}\label{df:responsibility gap free}
A law $L$ in a game $(\mathcal{A},\Delta,\mathbb{P})$ is \textbf{gap-free} if there is at least one responsible agent in each profile $\delta\in \mathbb{P}$.
\end{definition}

Observe that, by Lemma~\ref{lm:useful law} and statement~\eqref{eq:profile support set}, if a law is useful, then there is at least one legally responsible agent in each prohibited profile.
Thus, \textit{a useful law is also gap-free}.
Conversely, if a law is not useful, then there is at least one prohibited profile where no agent is legally responsible. 
To make such a law gap-free, there should be a principal agent who bears counterfactual responsibility in those lawful but prohibited profiles.
This intuition is formally captured by the next lemma. See Appendix~\ref{app_sec:proof of Lemma gap free law} for its formal proof.

\begin{lemma}\label{lm:gap free law}
A law $L$ in a game $(\mathcal{A},\Delta,\mathbb{P})$ is gap-free if and only if $L$ is useful or there is an agent $a\in\mathcal{A}$ and a safe action of agent $a$ in the law-imposed game $(\mathcal{A},\Delta^L,\mathbb{P}^L)$.
\end{lemma}

Now, we formalise the minimality of gap-freeness and gap-free reduction in a manner analogous to Section~\ref{sec:formalisation-useful}.

\begin{definition}\label{df:minimal gap free}
A law $L$ in a game is \textbf{minimal-gap-free} if $L$ is gap-free and no law $L'\subsetneq L$ is gap-free in the same game.
\end{definition}

\begin{definition}\label{df:minimum gap free reduction}
For a gap-free law $L$ in a game,
\begin{enumerate}
    \item a \textbf{gap-free reduction} of $L$ is a gap-free law $L'\subseteq L$;\label{dfitem:gap free reduction}
    \item a gap-free reduction $L'$ of $L$ is called \textbf{minimum} if there is no gap-free reduction $L''$ of $L$ such that $|L''|<|L'|$.\label{dfitem:minimum gap free reduction}
\end{enumerate}
\end{definition}

Recall that the law $\bigcup\Delta$ is useful and thus gap-free by Lemma~\ref{lm:gap free law}.
Hence, designing a gap-free law is equivalent to finding a gap-free reduction of the law $\bigcup\Delta$. Accordingly,
\textit{minimising an existing gap-free law encompasses the task of minimum gap-free law design}.

\subsection{Hypergraphs with Fixed Rank}\label{sec:formalisation-hypergraph}

Now, let us introduce the hypergraphs that are used later to resolve the law design problems.
Unlike standard graphs where each edge connects exactly two vertices, a hypergraph allows each edge (\textit{a.k.a.} hyperedge) to connect any positive number of vertices.
In particular, we consider hypergraphs where each edge contains \textit{at most} $k$ vertices for any fixed parameter $k\geq 1$, as formalised in the next definition.

\begin{definition}\label{df:k-rank hypergraph}
For any integer $k\geq 1$, a \textbf{rank-$k$ hypergraph} (\textit{abbr.} \textbf{$k$-graph}) is a tuple $(V,E)$ such that $V$ is a finite set of vertices and $E$ is a set of (hyper)edges where $e\subseteq V$ and $1\leq|e|\leq k$ for each edge $e\in E$.
\end{definition}

Note that an edge is a set of vertices.
A vertex is said to \textit{cover} an edge if the edge includes the vertex.
A {\bf\em vertex cover} of a $k$-graph is a set of vertices that \textit{intersects} with every edge in the $k$-graph, as formalised below.

\begin{definition}\label{df:vertex cover}
For a $k$-graph $(V,E)$,
\begin{enumerate}
    \item a set $C$ is called a \textbf{vertex cover} if $C\subseteq V$ and $C\cap e\neq\varnothing$ for each edge $e\in E$;\label{dfitem:vertex cover cover}
    \item a vertex cover $C$ is called \textbf{minimal} if there is no vertex cover $C'\subsetneq C$;\label{dfitem:vertex cover minimal}
    \item a vertex cover $C$ is called \textbf{minimum} if there is no vertex cover $C'$ such that $|C'|<|C|$.\label{dfitem:minimum vertex cover}
\end{enumerate}
\end{definition}
We consider the following problems about vertex cover, which are collectively referred to as $\mathbf{VC}$ problems:
\begin{itemize}
    \item $\mathbf{IsVC}$: to verify if a set is a vertex cover of a $k$-graph.
    \item $\mathbf{IsMiniVC}$: to verify if a set is a minimal vertex cover of a $k$-graph.
    \item $\mathbf{MinVC}$: to find a minimum vertex cover of a $k$-graph.
\end{itemize}
These problems are extensively explored in the literature.
In particular, $\mathbf{IsVC}$ and $\mathbf{IsMiniVC}$ can both be solved efficiently (\textit{i.e.} in polynomial time).
In contrast, $\mathbf{MinVC}$ is NP-hard \cite{garey1979computers}, which means no efficient algorithm is believed to exist for this problem.
As a compromise, {\em approximation algorithms} were developed to efficiently find \textit{good enough, though not necessarily optimal}, solutions for those hard problems.
An algorithm approximates $\mathbf{MinVC}$ within \textit{factor} $t$ if, for any $k$-graph, the size of the vertex cover it finds is at most $t$ times the size of the minimum vertex cover.
It is shown that $\mathbf{MinVC}$ can be \textit{effectively approximated within factor $k$}, using greedy algorithms \cite{bar1981linear,hall1986fast} or linear programming relaxation \cite{hochbaum1982approximation}. 
However, it is hard to do better than this \cite{holmerin2002improved,dinur2005new}. 
In what follows, we will rely on the following algorithms and theorem.

\begin{assumption}
    While solving the law design problems, the following efficient $\mathbf{VC}$ algorithms serve as {\bf\em gadgets}:
    \begin{itemize}
    \item $\mathtt{IsVC}$: an algorithm for $\mathbf{IsVC}$.
    \item $\mathtt{IsMiniVC}$: an algorithm for $\mathbf{IsMiniVC}$.
    \item $\mathtt{AppMinVC}$: a $k$-approximation algorithm for $\mathbf{MinVC}$.
\end{itemize}
\end{assumption}

\begin{theorem}[\citeauthor{khot2008vertex}, \citeyear{khot2008vertex}]\label{th:vertex cover hard to approximate}
    $\mathbf{MinVC}$ is NP-hard to approximate within factor $k-\epsilon$ for any $\epsilon\!>\!0$ when $k\!\geq\! 2$.\footnote{\citet{khot2008vertex} consider a subset of $k$-graphs where each edge has exactly $k$ vertices. The theorem holds under the Unique Game Conjecture \cite{khot2002power}.}
\end{theorem}

\section{Useful Law Design}\label{sec:useful law}

To continue the discussion in Section~\ref{sec:formalisation-useful}, we consider the next three problems about useful law design, which are collectively referred to as $\mathbf{UL}$ problems:
\begin{itemize}
    \item $\mathbf{IsUL}$: to verify if a set is a useful law in a game.
    \item $\mathbf{IsMiniUL}$: to verify if a set is a minimal-useful law in a game.
    \item $\mathbf{MinUR}$: to find a minimum useful reduction of a useful law in a game.
\end{itemize}
In particular, we establish an equivalence between the $\mathbf{VC}$ problems and the $\mathbf{UL}$ problems by providing two-way polynomial-time reductions.
On top of this, we illustrate a way to solve the $\mathbf{UL}$ problems using the $\mathbf{VC}$ algorithms.

\subsection{Reducing Vertex Cover to Useful Law}\label{sec:useful law-hardness}

In this subsection, we show that any instance of a $\mathbf{VC}$ problem can be reduced to an instance of a $\mathbf{UL}$ problem in polynomial time.
This demonstrates that the $\mathbf{UL}$ problems are \textit{at least as hard as} the $\mathbf{VC}$ problems.
The polynomial-time reduction is formally captured below.

\begin{definition}\label{df:graph to game}
For any $k$-graph $(V,E)$, let $\mathcal{G}_{(k,V,E)}$ be the game $(\mathcal{A}_{(k,V,E)},\Delta_{(k,V,E)},\mathbb{P}_{(k,V,E)})$ such that 
\begin{enumerate}
    \item $\mathcal{A}_{(k,V,E)}=[k]$ where $[k]=\{i\in\mathbb{N}\mid 1\leq i\leq k\}$;\label{dfitem:graph to game agents}
    \item $\Delta_{(k,V,E)}=\{\Delta_i\}_{i\in[k]}$ where $\Delta_i=V$ for each $i\in[k]$;\label{dfitem:graph to game actions}
    \item $\mathbb{P}_{(k,V,E)}=\{\delta^e=\{\delta^e_{i}\}_{i\in[k]}\mid e\in E\}$ where $\delta^e_i$ is the $((i\ \mathrm{mod}\ |e|)\!+\!1)^{\text{th}}$ item in any predefined order of set $e$.\label{dfitem:graph to game prohibition}
\end{enumerate}
\end{definition}

Note that, in the above definition, $k\geq 1$ and set $V$ is finite by Definition~\ref{df:k-rank hypergraph}. Thus, $\mathcal{G}_{(k,V,E)}$ is well-defined by Definition~\ref{df:game}.
Informally, $\mathcal{G}_{(k,V,E)}$ is a game of $k$ agents with the same action space $V$. 
Each vertex in set $V$ is an action. 
Each edge $e\in E$ corresponds to a prohibited profile $\delta^e$ where every vertex in edge $e$ is taken as an action $\delta^e_i$ of some agent $i$.
Then, by Lemma~\ref{lm:useful law} and item~\ref{dfitem:vertex cover cover} of Definition~\ref{df:vertex cover}, it is easy to verify the next theorem.
See Appendix~\ref{app_sec:proof of Theorem vertex cover=useful} for its proof.

\begin{theorem}\label{th:vertex cover=useful}
A set $C$ is a vertex cover of a $k$-graph $(V,E)$ if and only if $C$ is a useful law in the game $\mathcal{G}_{(k,V,E)}$.
\end{theorem}

Note that, due to the consistency of minimality in item~\ref{dfitem:vertex cover minimal} of Definition~\ref{df:vertex cover} and Definition~\ref{df:minimal useful law}, Theorem~\ref{th:vertex cover=useful} indeed implies reductions from $\mathbf{IsVC}$ and $\mathbf{IsMiniVC}$ to $\mathbf{IsUL}$ and $\mathbf{IsMiniUL}$, respectively.
On the other hand, the set $V=\bigcup\Delta_{(k,V,E)}$ is a useful law in the game $\mathcal{G}_{(k,V,E)}$ by Definition~\ref{df:graph to game} and Lemma~\ref{lm:useful law}.
Then, every useful law in the game $\mathcal{G}_{(k,V,E)}$ is a useful reduction of law $V$ by item~\ref{dfitem:useful reduction} of Definition~\ref{df:minimum useful reduction}.
Thus, the next corollary follows from Theorem~\ref{th:vertex cover=useful}.

\begin{corollary}\label{cr:vertex cover=useful reduction}
A set $C$ is a vertex cover of a $k$-graph $(V,E)$ if and only if $C$ is a useful reduction of law $V$ in game $\mathcal{G}_{(k,V,E)}$.
\end{corollary}

Note that, due to the consistency of minimality in item~\ref{dfitem:minimum vertex cover} of Definition~\ref{df:vertex cover} and item~\ref{dfitem:minimum useful reduction} of Definition~\ref{df:minimum useful reduction}, Corollary~\ref{cr:vertex cover=useful reduction} implies a reduction from $\mathbf{MinVC}$ to $\mathbf{MinUR}$.
This, together with Theorem~\ref{th:vertex cover hard to approximate}, further implies an \textit{inapproximability result} of the $\mathbf{MinUR}$ problem as stated in the next theorem. 
See Appendix~\ref{app_sec:proof of Theorem MiUR hard to approximate} for its proof.

\begin{theorem}\label{th:MinUR hard to approximate}
    $\mathbf{MinUR}$ in a game $(\mathcal{A},\Delta,\mathbb{P})$ is NP-hard to approximate within factor $|\mathcal{A}|-\epsilon$ for any $\epsilon\!>\!0$ when $|\mathcal{A}|\!\geq\! 2$.
\end{theorem}

\subsection{Reducing Useful Law to Vertex Cover}\label{sec:useful law-solution}

In this subsection, we show that any instance of a $\mathbf{UL}$ problem can be polynomially reduced to an instance of a $\mathbf{VC}$ problem.
This demonstrates that the $\mathbf{UL}$ problems are \textit{no harder than} the $\mathbf{VC}$ problems and suggests how $\mathbf{VC}$ algorithms can be used to solve the $\mathbf{UL}$ problems.

Technically, for any game $(\mathcal{A},\Delta,\mathbb{P})$, consider the $|\mathcal{A}|$-graph $(\bigcup\Delta,\mathcal{S}(\mathbb{P}))$ where $\mathcal{S}(\mathbb{P})=\{\mathcal{S}(\delta)\mid \delta\in\mathbb{P}\}$.\footnote{
$(\bigcup\Delta,\mathcal{S}(\mathbb{P}))$ is an $|\mathcal{A}|$-graph by Lemma~\ref{lm:game to graph} in Appendix~\ref{app_sec:game to |A|-graph}.
}
It is not hard to get the next theorem by Lemma~\ref{lm:useful law} and item~\ref{dfitem:vertex cover cover} of Definition~\ref{df:vertex cover}. See Appendix~\ref{app_sec:proof of Theorem useful=vertex cover} for its formal proof.

\begin{theorem}\label{th:useful=vertex cover}
A set $L$ is a useful law in a game $(\mathcal{A},\Delta,\mathbb{P})$ if and only if $L$ is a vertex cover of the graph $(\bigcup\Delta,\mathcal{S}(\mathbb{P}))$.
\end{theorem}

Opposite to Theorem~\ref{th:vertex cover=useful}, Theorem~\ref{th:useful=vertex cover} implies reductions from $\mathbf{IsUL}$ and $\mathbf{IsMiniUL}$ to $\mathbf{IsVC}$ and $\mathbf{IsMiniVC}$, respectively.
The corresponding algorithms $\mathtt{IsUL}$ and $\mathtt{IsMiniUL}$ call $\mathtt{IsVC}$ and $\mathtt{IsMiniVC}$, respectively, as shown in Algorithm~\ref{alg:UL design} of Appendix~\ref{app_sec:algorithms for the UL problems}.

Next, we consider the $\mathbf{MinUR}$ problem.
Note that, by Theorem~\ref{th:useful=vertex cover}, reducing a useful law corresponds to finding a smaller vertex cover within an existing one.
Meanwhile, the smaller vertex cover intersects every edge at the vertices in the original cover.
Given this observation, the smaller vertex cover can be regarded as a vertex cover in a subgraph induced by the original cover, which is formalised as follows.

\begin{definition}\label{df:vertex set reduced graph}
For a vertex cover $C$ of a $k$-graph $(V,E)$, the induced subgraph is the $k$-graph $(C,E^C)$ such that $E^C=\{C\cap e\mid e\in E\}$.
\end{definition}

The subgraph $(C,E^C)$ is well-defined because $C\cap e\neq\varnothing$ and $|C\cap e|\leq |e|\leq k$ for each edge $e\in E$ as $C$ is a vertex cover in the $k$-graph $(V,E)$. Moreover, observe that a vertex cover of the subgraph $(C,E^C)$ is also a vertex cover of the graph $(V,E)$.
Then, Theorem~\ref{th:useful=vertex cover} implies the next theorem. See Appendix~\ref{app_sec:proof of Theorem useful reduction=vertex cover} for its formal proof.

\begin{theorem}\label{th:useful reduction=vertex cover}
For a useful law $L$ in a game $(\mathcal{A},\Delta,\mathbb{P})$, law $L'$ is a useful reduction of $L$ if and only if $L'$ is a vertex cover of the induced subgraph $(L,\mathcal{S}(\mathbb{P})^L)$.
\end{theorem}


In contrast with Corollary~\ref{cr:vertex cover=useful reduction},
Theorem~\ref{th:useful reduction=vertex cover} implies a reduction from $\mathbf{MinUR}$ to $\mathbf{MinVC}$ and an algorithm $\mathtt{AppMinUR}$ that calls $\mathtt{AppMinVC}$ (see Algorithm~\ref{alg:UL design} of Appendix~\ref{app_sec:algorithms for the UL problems}).
Note that the graph $(L,\mathcal{S}(\mathbb{P})^L)$ in Theorem~\ref{th:useful reduction=vertex cover} is an $|\mathcal{A}|$-graph, on which $\mathtt{AppMinVC}$ is an  $|\mathcal{A}|$-approximation.
This makes $\mathtt{AppMinUR}$ an $|\mathcal{A}|$-approximation of $\mathbf{MinUR}$, where $|\mathcal{A}|$ is the number of agents in the input game (see Appendix~\ref{app_sec:MinUR approximation factor} for a discussion). It means $\mathtt{AppMinUR}$ achieves a nearly optimal approximation factor by Theorem~\ref{th:MinUR hard to approximate}.

\section{Gap-Free Law Design}\label{sec:gap-free law}

Following Section~\ref{sec:formalisation-gap free}, in this section, we consider the next three problems about gap-free law design, which are collectively referred to as $\mathbf{GFL}$ problems:
\begin{itemize}
    \item $\mathbf{IsGFL}$: to verify if a set is a gap-free law in a game.
    \item $\mathbf{IsMiniGFL}$: to verify if a set is a minimal-gap-free law in a game.
    \item $\mathbf{MinGFR}$: to find a minimum gap-free reduction of a gap-free law in a game.
\end{itemize}

We first show the hardness of the $\mathbf{GFL}$ problems by reductions from the $\mathbf{UL}$ problems. Then, we provide a way to solve the $\mathbf{GFL}$ problems again using the $\mathbf{VC}$ algorithms.

\subsection{Reducing Useful Law to Gap-Free Law}\label{sec:gap-free law-hardness}

In this subsection, we show that any instance of a $\mathbf{UL}$ problem can be polynomially reduced to an instance of a $\mathbf{GFL}$ problem. 
Specifically, for any game $(\mathcal{A},\Delta,\mathbb{P})$, we construct a game $(\bar{\mathcal{A}},\bar{\Delta},\bar{\mathbb{P}})$ as illstrated in the next definition.

\begin{definition}\label{df:useful game to gap free game}
For a game $(\mathcal{A},\Delta,\mathbb{P})$, an agent $\gamma\notin\mathcal{A}$, and two distinct actions $p,n\notin\bigcup\Delta$, let the game $(\bar{\mathcal{A}},\bar{\Delta},\bar{\mathbb{P}})$ be:
\begin{enumerate}
    \item $\bar{\mathcal{A}}=\mathcal{A}\cup\{\gamma\}$;\label{dfitem:useful game to gap free game agent}
    \item $\bar{\Delta}=\{\bar{\Delta}_a\}_{a\in\bar{\mathcal{A}}}$ where 
    $
        \bar{\Delta}_a=\begin{cases}
            \Delta_a\cup\{n\}, &\!\text{if } a\in\mathcal{A};\\
            \{p, n\}, &\!\text{if } a=\gamma;
        \end{cases}
    $\label{dfitem:useful game to gap free game action}
    \item $\bar{\mathbb{P}}=\bar{\mathbb{P}}_1\cup\bar{\mathbb{P}}_2\cup\bar{\mathbb{P}}_3$ where\begin{itemize}
        \item $\bar{\mathbb{P}}_1=
        \{\bar{\delta}\mid \bar{\delta}_{\gamma}=p, (\exists\delta\in\mathbb{P},\forall a\in\mathcal{A},\bar{\delta}_a=\delta_a)\}$;
        \item $\bar{\mathbb{P}}_2=\{\bar{\delta}\mid\exists a\in\mathcal{A}\,(\bar{\delta}_a\in\Delta_a, \forall b\in\bar{\mathcal{A}}\setminus\{a\},\bar{\delta}_b=n)\}$;
        \item $\bar{\mathbb{P}}_3=\{\bar{\delta}\mid \forall a\in\bar{\mathcal{A}}, \bar{\delta}_a=n\}$.
    \end{itemize} \label{dfitem:useful game to gap free game prohibition}
\end{enumerate}
\end{definition}

Observe that the new agent $\gamma$ takes action $n$ in each profile in sets $\bar{\mathbb{P}}_2,\bar{\mathbb{P}}_3$.
Meanwhile, by Lemma~\ref{lm:useful law} and the definition of $\bar{\mathbb{P}}_1$, a useful law $L$ in the game $(\mathcal{A},\Delta,\mathbb{P})$ intersects with each profile in set $\bar{\mathbb{P}}_1$. 
Then, law $L$ makes $p$ a safe action of agent $\gamma$ in the game $(\bar{\mathcal{A}},\bar{\Delta}^L,\bar{\mathbb{P}}^L)$ by Lemma~\ref{lm:safe action in law imposed game}.
Thus, law $L$ is a gap-free law in the game $(\bar{\mathcal{A}},\bar{\Delta},\bar{\mathbb{P}})$ by Lemma~\ref{lm:gap free law}.
Theorem~\ref{th:reduce useful to gap-free} below formalises the above observation and establishes its converse as well.
See Appendix~\ref{app_sec:proof of Theorem reduce useful to gap-free} for its formal proof.

\begin{theorem}\label{th:reduce useful to gap-free}
A set $L\subseteq\bigcup\Delta$ is a useful law in a game $(\mathcal{A},\Delta,\mathbb{P})$ if and only if $L$ is a gap-free law in the game $(\bar{\mathcal{A}},\bar{\Delta},\bar{\mathbb{P}})$.
\end{theorem}

Note that, due to the consistency of minimality in Definition~\ref{df:minimal useful law} and Definition~\ref{df:minimal gap free}, Theorem~\ref{th:reduce useful to gap-free} indeed implies reductions from $\mathbf{IsUL}$ and $\mathbf{IsMiniUL}$ to $\mathbf{IsGFL}$ and $\mathbf{IsMiniGFL}$, respectively.
Moreover, when considering a reduction $L'$ of a useful law $L$ in the game $(\mathcal{A},\Delta,\mathbb{P})$, it is guaranteed that $L'\subseteq L\subseteq\bigcup\Delta$.
Then, by item~\ref{dfitem:useful reduction} of  Definition~\ref{df:minimum useful reduction} and item~\ref{dfitem:gap free reduction} of Definition~\ref{df:minimum gap free reduction}, Theorem~\ref{th:reduce useful to gap-free} further implies the next corollary.

\begin{corollary}\label{cr:useful reduction to gap-free reduction}
For a useful law $L$ in a game $(\mathcal{A},\Delta,\mathbb{P})$, a set $L'$ is a useful reduction of $L$ in the game $(\mathcal{A},\Delta,\mathbb{P})$ if and only if $L'$ is a gap-free reduction of $L$ in the game $(\bar{\mathcal{A}},\bar{\Delta},\bar{\mathbb{P}})$.
\end{corollary}

Given the consistency of minimality in item~\ref{dfitem:minimum useful reduction} of  Definition~\ref{df:minimum useful reduction} and item~\ref{dfitem:minimum gap free reduction} of Definition~\ref{df:minimum gap free reduction}, Corollary~\ref{cr:useful reduction to gap-free reduction} implies a reduction from $\mathbf{MinUR}$ to $\mathbf{MinGFR}$.
This, together with Theorem~\ref{th:MinUR hard to approximate} and the extra agent $\gamma$ in Definition~\ref{df:useful game to gap free game}, further implies an \textit{inapproximability result} of  $\mathbf{MinGFR}$ as stated in the next theorem. See Appendix~\ref{app_sec:proof of Theorem MinGFR hard to approximate} for its proof.

\begin{theorem}\label{th:MinGFR hard to approximate}
    $\mathbf{MinGFR}$ in a game $(\mathcal{A},\Delta,\mathbb{P})$ is NP-hard to approximate within factor $|\mathcal{A}|-1-\epsilon$ for any $\epsilon\!>\!0$ when $|\mathcal{A}|\!\geq\! 3$.
\end{theorem}

\subsection{Reducing Gap-Free Law to Vertex Cover}

In this subsection, we reduce the $\mathbf{GFL}$ problems to the $\mathbf{VC}$ problems.
By this means, we show a way to address the $\mathbf{GFL}$ problems using the $\mathbf{VC}$ algorithms.

Recall Lemma~\ref{lm:gap free law} that the gap-freeness of a law corresponds to the usefulness of the law and the existence of a safe action in the law-imposed game.
Section~\ref{sec:useful law-solution} demonstrated how the $\mathbf{UL}$ problems can be addressed using the $\mathbf{VC}$ algorithms.
Now, we discuss how a ``safe action in a law-imposed game'' is captured in the $\mathbf{VC}$ context.
Note that not every action can become a safe action in a law-imposed game. If it can, we say the action is {\bf\em safable}.

\begin{definition}\label{df:safable action}
For a game $(\mathcal{A},\Delta,\mathbb{P})$, an action $d\in\bigcup\Delta$ is \textbf{safable} if there is a law $L$ and an agent $a$ such that $d$ is a safe action of agent $a$ in the law-imposed game $(\mathcal{A},\Delta^L,\mathbb{P}^L)$.    
\end{definition}

The next lemma characterises when an action is safable: every agent taking it does not lead to a prohibited outcome. See Appendix~\ref{app_sec:proof of Lemma safable action} for its formal proof.

\begin{lemma}\label{lm:safable action}
For a game $(\mathcal{A},\Delta,\mathbb{P})$, an action $d\in\bigcup\Delta$ is safable if and only if $\mathcal{S}(\delta)\neq\{d\}$ for each profile $\delta\in\mathbb{P}$.
\end{lemma}

Next, let us recall Lemma~\ref{lm:safe action in law imposed game}, which characterises when an action $d$ is a safe action of agent $a$ in a law-imposed game $(\mathcal{A},\Delta^L,\mathbb{P}^L)$.
Note that the second half of Lemma~\ref{lm:safe action in law imposed game} implies that $L\subseteq\bigcup\Delta\setminus\{d\}$ and $L\cap(\mathcal{S}(\delta)\setminus\{d\})\neq\varnothing$ for each profile $\delta\in\mathbb{P}$ such that $\delta_a=d$.
This, by item~\ref{dfitem:vertex cover cover} of Definition~\ref{df:vertex cover}, implies that $L$ is a vertex cover in the graph defined below.

\begin{definition}\label{df:safe action verify graph}
For a game $(\mathcal{A},\Delta,\mathbb{P})$, an agent $a\in\mathcal{A}$, and a safable action $d\in\Delta_a$, let the $|\mathcal{A}|$-graph $\mathcal{H}^{a,d}_{(\mathcal{A},\Delta,\mathbb{P})}$ be the pair $(\mathcal{V}^{a,d}_{(\mathcal{A},\Delta,\mathbb{P})},\mathcal{E}^{a,d}_{(\mathcal{A},\Delta,\mathbb{P})})$ where
\begin{enumerate}
    \item $\mathcal{V}^{a,d}_{(\mathcal{A},\Delta,\mathbb{P})}=(\bigcup\Delta)\setminus\{d\}$;\label{dfitem:safe action verify graph vertex}
    \item $\mathcal{E}^{a,d}_{(\mathcal{A},\Delta,\mathbb{P})}=\{\mathcal{S}(\delta)\setminus\{d\}\mid \delta\in\mathbb{P}, \delta_a=d\}$.\label{dfitem:safe action verify graph edge}
\end{enumerate}
\end{definition}

Note that the $|\mathcal{A}|$-graph $\mathcal{H}^{a,d}_{(\mathcal{A},\Delta,\mathbb{P})}$ is well-defined because $1\!\leq\!|\mathcal{S}(\delta)\setminus\{d\}|\!<\!|\mathcal{S}(\delta)|\!\leq\! |\mathcal{A}|$ by Lemma~\ref{lm:safable action} and statement~\eqref{eq:profile support set}.
Following the above observation, the next lemma bridges ``a safe action in a law-imposed game'' with the $\mathbf{VC}$ problems. See Appendix~\ref{app_sec:proof of Lemma safe action=vertex cover in reduced graph} for its formal proof.

\begin{lemma}\label{lm:safe action=vertex cover in reduced graph}
For a law $L$ in a game $(\mathcal{A},\Delta,\mathbb{P})$ and an agent $a\in\mathcal{A}$, an action $d\in\bigcup\Delta$ is a safe action of agent $a$ in the law-imposed game $(\mathcal{A},\Delta^L,\mathbb{P}^L)$ if and only if $d\in\Delta_a^L$, $d$ is safable in the game $(\mathcal{A},\Delta,\mathbb{P})$, and $L$ is a vertex cover in the graph $\mathcal{H}^{a,d}_{(\mathcal{A},\Delta,\mathbb{P})}$.
\end{lemma}

Observe that Lemma~\ref{lm:safe action=vertex cover in reduced graph}, Theorem~\ref{th:useful=vertex cover}, and Lemma~\ref{lm:gap free law} imply Theorem~\ref{th:reduce gap-free to multiple vertex cover} below,
which further implies a Cook reduction from $\mathbf{IsGFL}$ to $\mathbf{IsVC}$ and a corresponding algorithm $\mathtt{IsGFL}$ that \textit{iteratively} calls $\mathtt{IsVC}$ for polynomial times, as illustrated in Algorithm~\ref{alg:GFL design} of Appendix~\ref{app_sec:algorithms for GFL problems}.

\begin{theorem}\label{th:reduce gap-free to multiple vertex cover}
A law $L$ in a game $(\mathcal{A},\Delta,\mathbb{P})$ is gap-free if and only if at least one of the following statements is true:
\begin{enumerate}
    \item $L$ is a vertex cover in the graph $(\bigcup\Delta,\mathcal{S}(\mathbb{P}))$;\label{thitem:reduce gap-free to multiple vertex cover 1}
    \item there is an agent $a\in\mathcal{A}$ and a safable action $d\in\Delta_a^L$ such that $L$ is a vertex cover in the graph $\mathcal{H}^{a,d}_{(\mathcal{A},\Delta,\mathbb{P})}$.\label{thitem:reduce gap-free to multiple vertex cover 2}
\end{enumerate}
\end{theorem}

Note that, by Definition~\ref{df:minimal gap free} and item~\ref{dfitem:gap free reduction} of Definition~\ref{df:minimum gap free reduction}, a gap-free law is not minimal if there is a ``strict'' gap-free reduction. 
By Lemma~\ref{lm:gap free law}, such a reduction retains the gap-freeness in three possible approaches: maintaining usefulness, maintaining a safe action under the original law, or introducing a new safe action.
Moreover, we can use an induced subgraph in Definition~\ref{df:vertex set reduced graph} to capture a reduction of an existing law.
Following the above hints, we can get the next two theorems. See Appendices~\ref{app_sec:proof Theorem reduce minimal gap-free to multiple minimal vertex cover} and \ref{app_sec:proof of Theorem reduce gap-free reduction to vertex cover in reduced graph} for their proofs.

\begin{theorem}\label{th:reduce minimal gap-free to multiple minimal vertex cover}
A gap-free law $L$ in a game $(\mathcal{A},\Delta,\mathbb{P})$ is minimal if and only if all of the following statements are true:
\begin{enumerate}
    \item if $L$ is a vertex cover of the graph $(\bigcup\Delta,\mathcal{S}(\mathbb{P}))$, then $L$ is a minimal vertex cover in this graph;\label{thitem:reduce minimal gap-free to multiple minimal vertex cover 1}
    \item for each agent $a\in\mathcal{A}$ and each safable action $d\in\Delta_a^L$, if $L$ is a vertex cover in the graph $\mathcal{H}^{a,d}_{(\mathcal{A},\Delta,\mathbb{P})}$, then $L$ is a minimal vertex cover in this graph;\label{thitem:reduce minimal gap-free to multiple minimal vertex cover 2}
    \item for each agent $a\in\mathcal{A}$ and each safable action $d\!\in\!\Delta_a\!\cap\! L$, the set $L\setminus\{d\}$ is not a vertex cover in graph $\mathcal{H}^{a,d}_{(\mathcal{A},\Delta,\mathbb{P})}$.\label{thitem:reduce minimal gap-free to multiple minimal vertex cover 3}
\end{enumerate}
\end{theorem}

\begin{theorem}\label{th:reduce gap-free reduction to vertex cover in reduced graph}
For a gap-free law $L$ in a game $(\mathcal{A},\Delta,\mathbb{P})$, a law $L'$ is a gap-free reduction of $L$ if and only if at least one of the following statements is true:
\begin{enumerate}
    \item $L$ is a vertex cover of the graph $(\bigcup\Delta,\mathcal{S}(\mathbb{P}))$ and $L'$ is a vertex cover in the subgraph $(L,\mathcal{S}(\mathbb{P})^L)$;\label{thitem:reduce gap-free reduction to vertex cover in reduced graph 1}
    \item there is an agent $a\in\mathcal{A}$ and a safable action $d\in\Delta_a^L$ such that $L$ is a vertex cover in graph $\mathcal{H}^{a,d}_{(\mathcal{A},\Delta,\mathbb{P})}$ and $L'$ is a vertex cover in the subgraph $(L,(\mathcal{E}^{a,d}_{(\mathcal{A},\Delta,\mathbb{P})})^L)$;\label{thitem:reduce gap-free reduction to vertex cover in reduced graph 2}
    \item there is an $a\!\in\!\mathcal{A}$ and a safable action $d\!\in\!\Delta_a\!\cap\! L$ such that $L\!\setminus\!\{d\}$ is a vertex cover in graph $\mathcal{H}^{a,d}_{(\mathcal{A},\Delta,\mathbb{P})}$ and $L'$ is a vertex cover in the subgraph $(L\!\setminus\!\{d\},(\mathcal{E}^{a,d}_{(\mathcal{A},\Delta,\mathbb{P})})^{L\setminus\{d\}})$.\label{thitem:reduce gap-free reduction to vertex cover in reduced graph 3}
\end{enumerate}
\end{theorem}

The above two theorems imply Cook reductions from $\mathbf{IsMiniGFL}$ and $\mathbf{MinGFR}$ to the $\mathbf{VC}$ problems. The corresponding algorithms $\mathtt{IsMiniGFL}$ and $\mathtt{AppMinGFR}$ call the $\mathbf{VC}$ algorithms polynomial times (see Algorithm~\ref{alg:GFL design} of Appendix~\ref{app_sec:algorithms for GFL problems}).
Note that all graphs in Theorem~\ref{th:reduce gap-free reduction to vertex cover in reduced graph} are $|\mathcal{A}|$-graphs, on which $\mathtt{AppMinVC}$ is an $|\mathcal{A}|$-approximation.
This makes $\mathtt{AppMinGFR}$ an $|\mathcal{A}|$-approximation of $\mathbf{MinGFR}$, where $|\mathcal{A}|$ is the number of agents in the input game (see Appendix~\ref{app_sec:MinGFR approximation factor} for a discussion).

\section{Conclusion}

For the usefulness and gap-freeness of law, we studied the corresponding law design problems by relating them to vertex cover problems in hypergraphs. 
We proved that the task of minimising a law while keeping its usefulness or gap-freeness is NP-hard even to approximate. 
We also proposed reductions from law design problems to vertex cover problems, which imply law-design algorithms that make polynomial-time calls to the vertex cover algorithms.

As an ending discussion, note that we don't consider the weight of actions in law minimisation. However, this can be done with no extra effort because the vertex cover algorithms in the literature apply to weighted vertices.

\bibliography{this,fromConfirmationReview}

\clearpage

\appendix

\begin{center}

{\bf \Large 
Technical Appendix
}


\end{center}

\section{Supplementary to Section~\ref{sec:formalisation}}\label{app_sec:formalisation}

\subsection{Proof of Lemma~\ref{lm:useful law}}\label{app_sec:proof of Lemma useful law}

\noindent\textbf{Lemma~\ref{lm:useful law}.} 
\textit{A law $L$ in a game $(\mathcal{A},\Delta,\mathbb{P})$ is useful if and only if $L\cap\mathcal{S}(\delta)\neq\varnothing$ for each profile $\delta\in\mathbb{P}$.}
\begin{proof}
By Definition~\ref{df:useful law}, it suffices to show that $\mathbb{P}^L=\varnothing$ in the law-imposed game if and only if $L\cap\mathcal{S}(\delta)\neq\varnothing$ for each profile $\delta\in\mathbb{P}$.

($\Rightarrow$)
Suppose there is a profile $\delta\in\mathbb{P}$ such that $L\cap\mathcal{S}(\delta)=\varnothing$.
Then, $\delta_a\notin L$ for each agent $a\in\mathcal{A}$ by statement~\eqref{eq:profile support set}.
Thus, $\delta\in\prod\Delta^L$ by item~\ref{dfitem:law-imposed game 1} of Definition~\ref{df:law-imposed game}.
Hence, $\delta\in\mathbb{P}^L$ by item~\ref{dfitem:law-imposed game 2} of Definition~\ref{df:law-imposed game} and the case assumption $\delta\in\mathbb{P}$.
Therefore, $\mathbb{P}^L\neq\varnothing$.

($\Leftarrow$)
Suppose $\mathbb{P}^L\neq\varnothing$. Then, there is a profile $\delta\in\mathbb{P}^L$. Thus, $\delta\in\prod\Delta^L$ by item~\ref{dfitem:law-imposed game 2} of Definition~\ref{df:law-imposed game}.
Hence, $\delta_a\notin L$ for each agent $a\in\mathcal{A}$ by item~\ref{dfitem:law-imposed game 1} of Definition~\ref{df:law-imposed game}.
Therefore, $L\cap\mathcal{S}(\delta)=\varnothing$ by statement~\eqref{eq:profile support set}.
\end{proof}

\subsection{Proof of Lemma~\ref{lm:safe action in law imposed game}}\label{app_sec:proof of Lemma safe action in law imposed game}

\noindent\textbf{Lemma~\ref{lm:safe action in law imposed game}. }
\textit{An action $d\in\bigcup\Delta$ is a safe action of an agent $a\in\mathcal{A}$ in the law-imposed game $(\mathcal{A},\Delta^L,\mathbb{P}^L)$ if and only if $d\in\Delta_a^L$ and $L\cap\mathcal{S}(\delta)\neq\varnothing$ for each $\delta\in\mathbb{P}$ such that $\delta_a=d$.}
\begin{proof}
By Definition~\ref{df:safe action} and Definition~\ref{df:law-imposed game}, action $d$ is a safe action of agent $a$ in the law-imposed game $(\mathcal{A},\Delta^L,\mathbb{P}^L)$ if and only if $d\in\Delta_a^L$ and $\delta_a\neq d$ for each profile $\delta\in\mathbb{P}^L$.
Then, to prove the statement of the lemma, it suffices to prove the equivalence between the next two statements:
\begin{inlinelist}[label=(\Roman*), ref=(\Roman*)]
    \item $\delta_a\neq d$ for each profile $\delta\in\mathbb{P}^L$;\label{listiem:6-30-1}
    \item $L\cap\mathcal{S}(\delta)\neq\varnothing$ for each profile $\delta\in\mathbb{P}$ such that $\delta_a=d$.\label{listiem:6-30-2}
\end{inlinelist}

\ref{listiem:6-30-1} $\Rightarrow$ \ref{listiem:6-30-2}: 
Suppose there exists a profile $\delta\in\mathbb{P}$ such that $\delta_a=d$ and $L\cap\mathcal{S}(\delta)=\varnothing$. 
Then, $\delta_b\notin L$ for each agent $b\in\!\mathcal{A}$ by statement~\eqref{eq:profile support set}.
Thus, $\delta$ is such that $\delta\in\mathbb{P}^L$ and $\delta_a=d$ by Definition~\ref{df:law-imposed game} and the assumptions $\delta\in\mathbb{P}$ and $\delta_a=d$.

\ref{listiem:6-30-2} $\Rightarrow$ \ref{listiem:6-30-1}: 
Suppose there is a profile $\delta\in\mathbb{P}^L$ such that $\delta_a=d$.
Then, $\delta\in\mathbb{P}$ and $\delta_b\notin L$ for each agent $b\in\mathcal{A}$ by Definition~\ref{df:law-imposed game}.
Thus, $\delta$ is such that $\delta\in\mathbb{P}$ and $L\cap\mathcal{S}(\delta)=\varnothing$ by statement~\eqref{eq:profile support set}.
\end{proof}

\subsection{Proof of Lemma~\ref{lm:gap free law}}\label{app_sec:proof of Lemma gap free law}

\noindent\textbf{Lemma~\ref{lm:gap free law}. }
\textit{A law $L$ in a game $(\mathcal{A},\Delta,\mathbb{P})$ is gap-free if and only if $L$ is useful or there is an agent $a\in\mathcal{A}$ and a safe action of agent $a$ in the law-imposed game $(\mathcal{A},\Delta^L,\mathbb{P}^L)$.}
\begin{proof}
($\Rightarrow$)
Suppose $L$ is gap-free but not useful.
Then, $\mathbb{P}^L\neq\varnothing$ by Definition~\ref{df:useful law}.
Consider an arbitrary profile $\delta\in\mathbb{P}^L$.
Then, $\delta\in\mathbb{P}$ and
\begin{equation}\label{eq:6-30-1}
    \text{$\delta_b\notin L$ for each agent $b\in\mathcal{A}$}
\end{equation}
by Definition~\ref{df:law-imposed game}.
Thus, there is an agent $a\in\mathcal{A}$ responsible for $\delta$ by Definition~\ref{df:responsibility gap free} and the assumption $L$ is gap-free. 
Then, there is a safe action $d$ of agent $a$ in the law-imposed game $(\mathcal{A},\Delta^L,\mathbb{P}^L)$ by Definition~\ref{df:responsibility} and statement~\eqref{eq:6-30-1}.

($\Leftarrow$)
Suppose $L$ is not gap-free. Then, by Definition~\ref{df:responsibility gap free},
\begin{equation}\label{eq:6-30-2}
    \text{there is a profile $\delta\in\mathbb{P}$ where no agent is responsible.}
\end{equation}
Consider an arbitrary agent $a\in\mathcal{A}$.
Then, $\delta_a\notin L$ by item~\ref{dfitem:legal responsibility} of Definition~\ref{df:responsibility} and statement~\eqref{eq:6-30-2}.
Thus, by statement~\eqref{eq:profile support set} and the arbitrariness of $a$,
\begin{equation}\textstyle\label{eq:6-30-3}
    \mathcal{S}(\delta)\cap L=\varnothing
\end{equation}
Hence, law $L$ is not useful by Lemma~\ref{lm:useful law} and that $\delta\in\mathbb{P}$ in statement~\eqref{eq:6-30-2}.
Meanwhile, by statements~\eqref{eq:6-30-2}, \eqref{eq:6-30-3}, and item~\ref{dfitem:counterfactual responsibility} of Definition~\ref{df:responsibility}, there is no safe action of (the arbitrary) agent $a$ in the law-imposed game $(\mathcal{A},\Delta^L,\mathbb{P}^L)$.
\end{proof}

\section{Supplementary to Section~\ref{sec:useful law}}\label{app_sec:useful law}

\subsection{Proof of Theorem~\ref{th:vertex cover=useful}}\label{app_sec:proof of Theorem vertex cover=useful}

Let us first illustrate Definition~\ref{df:graph to game} with an example as illustrated in Figure~\ref{fig:Reduction_VC2UL_graph}.
In this $3$-graph, there are five vertices $r,s,t,u,v$ and three edges: $e_1=\{s\}$, $e_2=\{t,u\}$, and $e_3=\{r,u,v\}$.
The induced game, by item~\ref{dfitem:graph to game agents} of Definition~\ref{df:graph to game}, is played within three agents: $1,2,3$.
Each agent, by item~\ref{dfitem:graph to game actions}  of Definition~\ref{df:graph to game}, has access to five actions: $r,s,t,u,v$.
Moreover, by item~\ref{dfitem:graph to game prohibition}  of Definition~\ref{df:graph to game}, there are three prohibited outcomes: $\delta^{e_1}$, $\delta^{e_2}$, $\delta^{e_3}$.
In particular, taking the predefined order in item~\ref{dfitem:graph to game prohibition} of  of Definition~\ref{df:graph to game} as the lexicographical order. Then, 
$
\delta^{e_1}_1=\delta^{e_1}_2=\delta^{e_1}_3=s
$;
$
\delta^{e_2}_1=\delta^{e_2}_3=u, \delta^{e_2}_2=t
$;
$
\delta^{e_3}_1=u, \delta^{e_3}_2=v, \delta^{e_3}_3=r
$.
The induced game is as illustrated in Figure~\ref{fig:Reduction_VC2UL_game}.

\begin{figure}[htb]
    \centering
    \scalebox{0.45}{\includegraphics{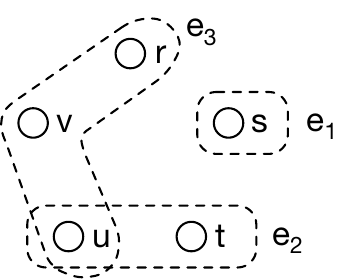}}
    \caption{A sample of $3$-graph.}
    \label{fig:Reduction_VC2UL_graph}
\end{figure}

\begin{figure*}[htb]
    \centering
    \scalebox{0.45}{\includegraphics{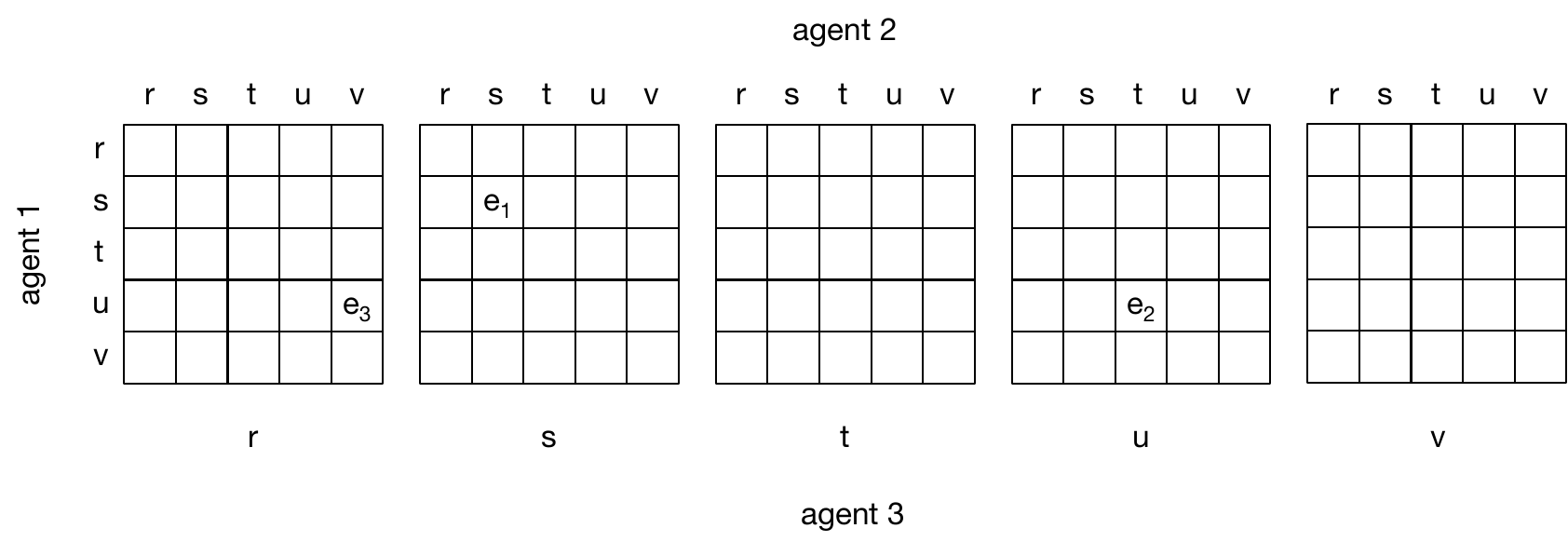}}
    \caption{The game induced from the $3$-graph in Figure~\ref{fig:Reduction_VC2UL_graph} by Definition~\ref{df:graph to game}.}
    \label{fig:Reduction_VC2UL_game}
\end{figure*}

Observe Figure~\ref{fig:Reduction_VC2UL_graph_vertexcover} that the set $\{s,u\}$ is a vertex cover of the graph in Figure~\ref{fig:Reduction_VC2UL_graph}.
If a law bans these two actions in the induced game in Figure~\ref{fig:Reduction_VC2UL_game}, then, as illustrated in Figure~\ref{fig:Reduction_VC2UL_game_usefullaw}, all outcomes corresponding to the greyed cells become unlawful, including all three prohibited outcomes.
In other words, all lawful profiles are non-prohibited. Therefore, the set $\{s,u\}$ is also a useful law in this game.

\begin{figure}[htb]
    \centering
    \scalebox{0.45}{\includegraphics{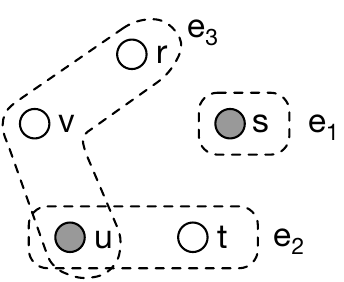}}
    \caption{Illustration of a vertex cover in the $3$-graph in Figure~\ref{fig:Reduction_VC2UL_graph}.}
    \label{fig:Reduction_VC2UL_graph_vertexcover}
\end{figure}

\begin{figure*}[htb]
    \centering
    \scalebox{0.45}{\includegraphics{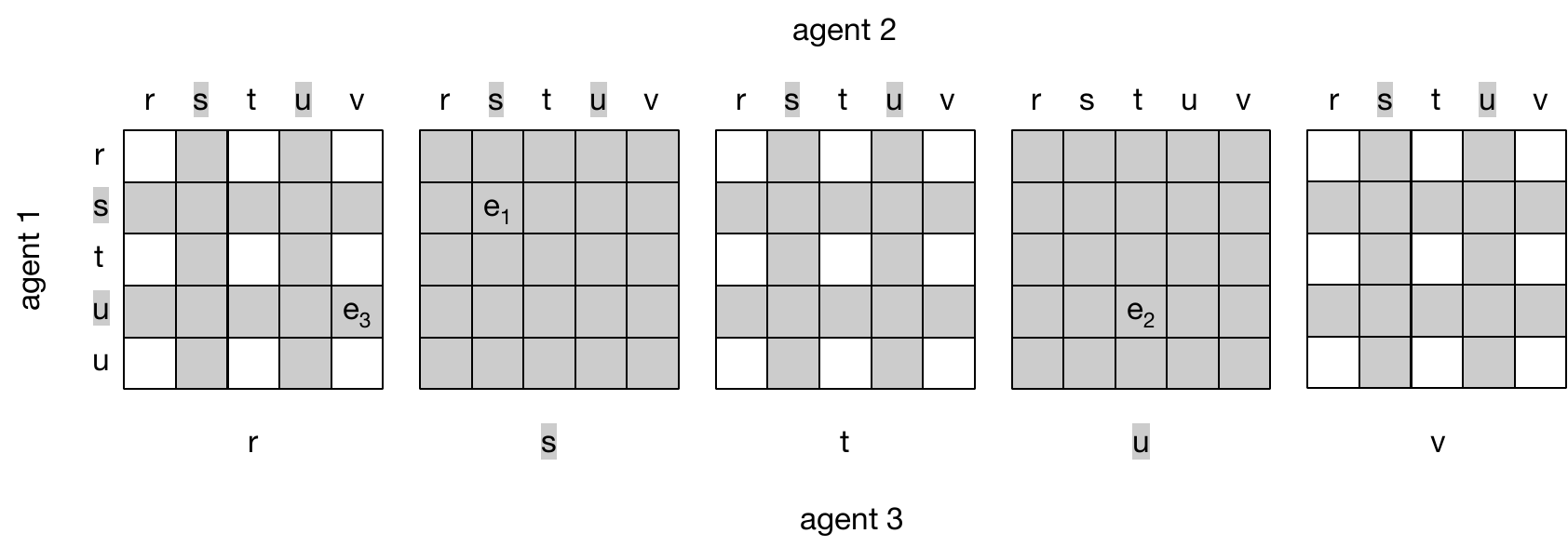}}
    \caption{Illustration of a useful law in the game in Figure~\ref{fig:Reduction_VC2UL_game}.}
    \label{fig:Reduction_VC2UL_game_usefullaw}
\end{figure*}

More generally, one can prove Theorem~\ref{th:vertex cover=useful} below, which bridges a vertex cover and a useful law.
To increase the readability of its proof, let us first prove the next lemma.

\begin{lemma}\label{lm:graph to game}
$\mathcal{S}(\delta^e)=e$ for each profile $\delta^e\in\mathbb{P}_{(k,V,E)}$.
\end{lemma}
\begin{proof}
($\subseteq$)
Consider an arbitrary action $d\in\mathcal{S}(\delta^e)$.
Then, $d=\delta^e_i$ for some agent $i\in[k]$ by statement~\eqref{eq:profile support set} and item~\ref{dfitem:graph to game agents} of Definition~\ref{df:graph to game}.
Thus, $d\in e$ by item~\ref{dfitem:graph to game prohibition} of Definition~\ref{df:graph to game}.

($\supseteq$)
Consider an arbitrary vertex $v\in e$.
Suppose $v$ is the $j^{\text{th}}$ item in the predefined order of set $e$ in item~\ref{dfitem:graph to game prohibition} of Definition~\ref{df:graph to game}.
Then,
\begin{equation}\label{eq:7-23-1}
    1\leq j\leq |e|.
\end{equation}
Note that $(V,E)$ is a $k$-graph by Definition~\ref{df:graph to game}. 
Then, $|e|\leq k$ by Definition~\ref{df:k-rank hypergraph}. 
Thus, $1\leq j\leq |e|\leq k$ by statement~\eqref{eq:7-23-1}.
Hence, one of the next cases must be true:

Case 1: $j=1$. Then, $|e|\ \mathrm{mod}\ |e|+1=j$. Thus, $\delta^e_{|e|}=v$ by item~\ref{dfitem:graph to game prohibition} of Definition~\ref{df:graph to game}. Then, $v\!\in\!\mathcal{S}(\delta^e)$ by statement~\eqref{eq:profile support set}.

Case 2: $1< j\leq |e|$. Then, $(j-1) \ \mathrm{mod}\ |e|+1=j$. Thus, $\delta^e_{j-1}=v$ by item~\ref{dfitem:graph to game prohibition} of Definition~\ref{df:graph to game}. Then, $v\!\in\!\mathcal{S}(\delta^e)$ by statement~\eqref{eq:profile support set}.
\end{proof}

Lemma~\ref{lm:graph to game} above implies that $L\cap e\neq\varnothing$ is true if and only if $L\cap\mathcal{S}(\delta^e)\neq\varnothing$.
This intuitively shows that a vertex set $L$ covers an edge $e$ if and only if the prohibited profile $\delta^e$ is unlawful under law $L$.
Next, we prove Theorem~\ref{th:vertex cover=useful}.

\noindent\textbf{Theorem~\ref{th:vertex cover=useful}. }
\textit{A set $C$ is a vertex cover of a $k$-graph $(V,E)$ if and only if $C$ is a useful law in the game $\mathcal{G}_{(k,V,E)}$.}
\begin{proof}
Note that, by item~\ref{dfitem:graph to game actions} of Definition~\ref{df:graph to game},
\begin{equation}\textstyle\label{eq:7-23-2}
    \bigcup\Delta_{(k,V,E)}=V.
\end{equation}

($\Rightarrow$)
Suppose set $C$ is a vertex cover of the graph $(V,E)$. Then, $C\subseteq V$ by item~\ref{dfitem:vertex cover cover} of Definition~\ref{df:vertex cover}.
Thus, set $C$ is a law in the game $\mathcal{G}_{(k,V,E)}$ by Definition~\ref{df:law in game} and statement~\eqref{eq:7-23-2}.
Hence, it suffices to show that $C\cap\mathcal{S}(\delta^e)\neq\varnothing$ for each profile $\delta^e\in\mathbb{P}_{(k,V,E)}$ by Lemma~\ref{lm:useful law}.

Suppose the opposite. Then, a profile $\delta^e\in\mathbb{P}_{(k,V,E)}$ such that $C\cap\mathcal{S}(\delta^e)=\varnothing$ exists.
Thus, $C\cap e=\varnothing$ by Lemma~\ref{lm:graph to game}.
Hence, $C$ is not a vertex cover of the graph $(V,E)$ by item~\ref{dfitem:vertex cover cover} of Definition~\ref{df:vertex cover}, which contradicts the ($\Rightarrow$) part assumption.

($\Leftarrow$)
Suppose $C$ is a useful law in the game $\mathcal{G}_{(k,V,E)}$.
Then, by Definition~\ref{df:law in game} and statement~\eqref{eq:7-23-2},
\begin{equation}\label{eq:7-23-3}
    C\subseteq V
\end{equation}
and $C\cap\mathcal{S}(\delta^e)\neq\varnothing$ for each profile $\delta^e\in\mathbb{P}_{(k,V,E)}$ by Lemma~\ref{lm:useful law}.
Thus, $C\cap e\neq\varnothing$ for each edge $e\in E$ by item~\ref{dfitem:graph to game prohibition} of Definition~\ref{df:graph to game} and Lemma~\ref{lm:graph to game}.
Hence, $C$ is vertex cover of the graph $(V,E)$ by statement~\eqref{eq:7-23-3} and item~\ref{dfitem:vertex cover cover} of Definition~\ref{df:vertex cover}.
\end{proof}

\subsection{Proof of Theorem~\ref{th:MinUR hard to approximate}}\label{app_sec:proof of Theorem MiUR hard to approximate}

\noindent\textbf{Theorem~\ref{th:MinUR hard to approximate}. }
\textit{$\mathbf{MinUR}$ in a game $(\mathcal{A},\Delta,\mathbb{P})$ is NP-hard to approximate within factor $|\mathcal{A}|-\epsilon$ for any $\epsilon\!>\!0$ when $|\mathcal{A}|\!\geq\! 2$.}
\begin{proof}
Suppose the opposite.
Then, there is an $\epsilon>0$ and an algorithm $\mathtt{Alg}$ that approximates $\mathbf{MinUR}$ with factor $|\mathcal{A}|-\epsilon$ for any game $(\mathcal{A},\Delta,\mathbb{P})$ where $|\mathcal{A}|\geq 2$.

Consider an arbitrary $k$-graph $(V,E)$ where $k\geq 2$. We consider the next two steps:
\begin{enumerate}
    \item compute the game $\mathcal{G}_{(k,V,E)}$ corresponding to the graph $(V,E)$ by Definition~\ref{df:graph to game};
    \item use the algorithm $\mathtt{Alg}$ to get an approximated useful reduction $L$ of law $V$ in the game $\mathcal{G}_{(k,V,E)}$.
\end{enumerate}
By step~2 above and Corollary~\ref{cr:vertex cover=useful reduction}, set $L$ is a vertex cover of the graph $(V,E)$.

Note that the minimum useful reduction of law $V$ in the game $\mathcal{G}_{(k,V,E)}$ is the minimum vertex cover of the graph $(V,E)$ by Corollary~\ref{cr:vertex cover=useful reduction}.
Meanwhile, the size of set $L$ is at most $(k-\epsilon)$ times the size of the minimum useful reduction of law $V$ by the assumption about the approximation factor of $\mathtt{Alg}$.
Thus, the size of set $L$ is at most $(k-\epsilon)$ times the size of the minimum vertex cover of the graph $(V,E)$.
Hence, by the arbitrariness of the $k$-graph $(V,E)$, the above two steps form an algorithm that approximates $\mathbf{MinVC}$ with factor $k-\epsilon$, which contradicts Theorem~\ref{th:vertex cover hard to approximate}.
\end{proof}

\subsection{the $|\mathcal{A}|$-graph $(\bigcup\Delta,\mathcal{S}(\delta))$}\label{app_sec:game to |A|-graph}

First, note that, as a standard notation in mathematics,
\begin{equation}\label{eq:game prohibition to graph edges}
 \mathcal{S}(\mathbb{P})=\{\mathcal{S}(\delta)\mid \delta\in\mathbb{P}\}.    
\end{equation}

We use the introductory example to illustrate an induced graph $(\bigcup\Delta,\mathcal{S}(\delta))$.
Note that the game in each three-day cycle of this example can be captured in Figure~\ref{fig:ND intro game matrix}.
In this game, there are three agents $a$, $b$, and $c$; each agent can choose to dump on the first, the second, or the third day in a three-day cycle.
The bio-hazard symbol \Biohazard\, captures the outcomes where the fish is killed, which happens when all three factories dump on the same day.
This setting contains $3\times 3$ relavant actions: $d_a^1$, $d_a^2$, $d_a^3$, $d_b^1$, $d_b^2$, $d_b^3$, $d_c^1$, $d_c^2$, $d_c^3$.
\begin{figure}[htb]
    \centering
    \scalebox{0.5}{\includegraphics{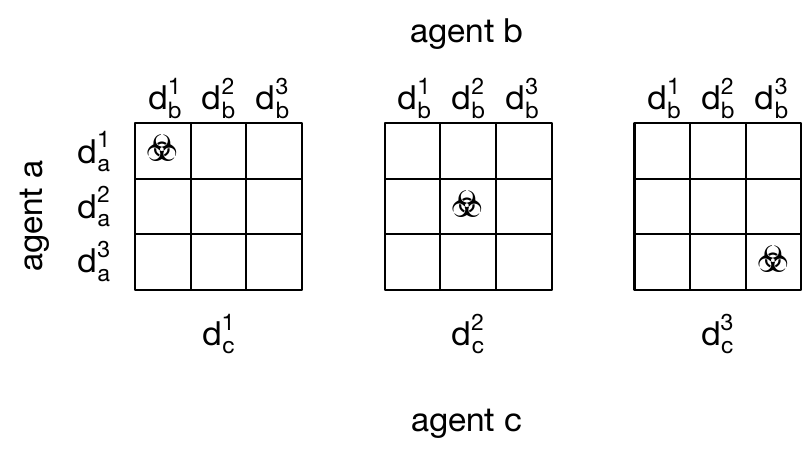}}
    \caption{Matrix game in accordance with the introductory example.}
    \label{fig:ND intro game matrix}
\end{figure}
The induced graph consists of nine vertices: $d_a^1$, $d_a^2$, $d_a^3$, $d_b^1$, $d_b^2$, $d_b^3$, $d_c^1$, $d_c^2$, $d_c^3$, as well as three edges: $e_1=\{d_a^1,d_b^1,d_c^1\}$, $e_2=\{d_a^2,d_b^2,d_c^2\}$, and $e_3=\{d_a^3,d_b^3,d_c^3\}$.
The induced graph is illustrated in Figure~\ref{fig:Reduction_UL2VC_originalgraph}.
\begin{figure}[htb]
    \centering
    \scalebox{0.5}{\includegraphics{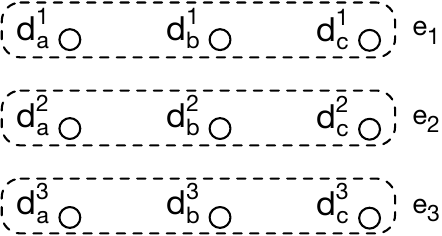}}
    \caption{The graph induced from the game in Figure~\ref{fig:ND intro game matrix}.}
    \label{fig:Reduction_UL2VC_originalgraph}
\end{figure}

Now, we formally prove that the tuple $(\bigcup\Delta,\mathcal{S}(\mathbb{P}))$ is an $|\mathcal{A}|$-graph.

\begin{lemma}\label{lm:game to graph}
For a game $(\mathcal{A},\Delta,\mathbb{P})$, the tuple $(\bigcup\Delta,\mathcal{S}(\mathbb{P}))$ is an $|\mathcal{A}|$-graph.
\end{lemma}
\begin{proof}
Note that, by items~\ref{dfitem:game agents} and \ref{dfitem:game actions} of Definition~\ref{df:game}, the sets $\mathcal{A}$ and $\Delta_a$ for each $a\in\mathcal{A}$ are all finite.
Then,
\begin{equation}\textstyle\label{eq:7-30-1}
    \text{the set $\bigcup\Delta$ is finite.}
\end{equation}

At the same time, $\mathcal{A}\neq\varnothing$ by item~\ref{dfitem:game agents} of Definition~\ref{df:game}.
Then, $1\leq|\mathcal{S}(\delta)|\leq|\mathcal{A}|$ for each profile $\delta\in\mathbb{P}$ by item~\ref{dfitem:game prohibitions} of Definition~\ref{df:game} and statement~\eqref{eq:profile support set}.
Thus, the statement of the lemma is true by Definition~\ref{df:k-rank hypergraph} and statement~\eqref{eq:7-30-1}.
\end{proof}

\subsection{Proof of Theorem~\ref{th:useful=vertex cover}}\label{app_sec:proof of Theorem useful=vertex cover}

Observe Figure~\ref{fig:Reduction_UL2VC_graphcoverL0} that the law $L_0$ in the introductory example is a vertex cover in the graph in Figure~\ref{fig:Reduction_UL2VC_originalgraph}.

\begin{figure}[htb]
    \centering
    \scalebox{0.5}{\includegraphics{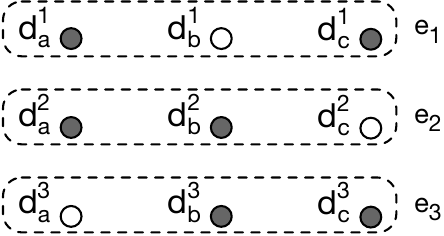}}
    \caption{$L_0$ is a useful law in the game in Figure~\ref{fig:ND intro game matrix} and a vertex cover in the graph in Figure~\ref{fig:Reduction_UL2VC_originalgraph}.}
    \label{fig:Reduction_UL2VC_graphcoverL0}
\end{figure}

More generally, we can prove the next theorem.

\noindent\textbf{Theorem~\ref{th:useful=vertex cover}. }
\textit{A set $L$ is a useful law in a game $(\mathcal{A},\Delta,\mathbb{P})$ if and only if $L$ is a vertex cover of the graph $(\bigcup\Delta,\mathcal{S}(\mathbb{P}))$.}
\begin{proof}
($\Rightarrow$)
Suppose that the set $L$ is a useful law in the game $(\mathcal{A},\Delta,\mathbb{P})$.
Then, $L\subseteq\bigcup\Delta$ and $\mathcal{S}(\delta)\cap L\neq\varnothing$ for each profile $\delta\in\mathbb{P}$ by Definition~\ref{df:law in game} and Lemma~\ref{lm:useful law}.
Thus, $L$ is a vertex cover of the graph $(\bigcup\Delta,\mathcal{S}(\mathbb{P}))$ by item~\ref{dfitem:vertex cover cover} of Definition~\ref{df:vertex cover} and statement~\eqref{eq:game prohibition to graph edges}. 

($\Leftarrow$) Suppose that the set $L$ is a vertex cover of the graph $(\bigcup\Delta,\mathcal{S}(\mathbb{P}))$.
Then, $L\subseteq\bigcup\Delta$ and $\mathcal{S}(\delta)\cap L\neq\varnothing$ for each profile $\delta\in\mathbb{P}$ by item~\ref{dfitem:vertex cover cover} of Definition~\ref{df:vertex cover} and statement~\eqref{eq:game prohibition to graph edges}.
Thus, $L$ is a useful law in the game $(\mathcal{A},\Delta,\mathbb{P})$ by Definition~\ref{df:law in game} and Lemma~\ref{lm:useful law}.
\end{proof}

\subsection{Proof of Theorem~\ref{th:useful reduction=vertex cover}}\label{app_sec:proof of Theorem useful reduction=vertex cover}

\begin{figure}[htb]
    \centering
    \scalebox{0.5}{\includegraphics{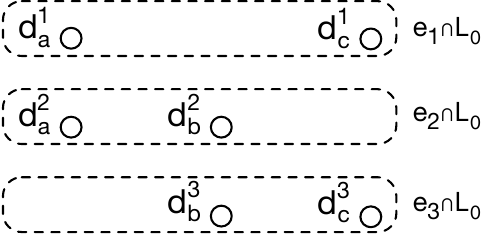}}
    \caption{An illustration of Definition~\ref{df:vertex set reduced graph} that uses the graph in Figure~\ref{fig:Reduction_UL2VC_originalgraph}.}
    \label{fig:Reduction_UL2VC_inducedgraph}
\end{figure}

Note that, by Theorem~\ref{th:useful=vertex cover}, reducing a useful law corresponds to finding a smaller vertex cover within an existing one.
For instance, the law $L_0$ forms a vertex cover in the graph in Figure~\ref{fig:Reduction_UL2VC_originalgraph}, as illustrated in Figure~\ref{fig:Reduction_UL2VC_graphcoverL0}.
To reduce law $L_0$ into a useful law is to select a subset of $L_0$ that can cover the graph in Figure~\ref{fig:Reduction_UL2VC_originalgraph}.
Observe that, this is equivalent to obtain a vertex cover in the graph shown in Figure~\ref{fig:Reduction_UL2VC_inducedgraph}.
In other words, the smaller vertex cover can be regarded as a vertex cover in a subgraph induced by the original cover, which is formalised as the theorem below.

\noindent\textbf{Theorem~\ref{th:useful reduction=vertex cover}. }
\textit{For a useful law $L$ in a game $(\mathcal{A},\Delta,\mathbb{P})$, law $L'$ is a useful reduction of $L$ if and only if $L'$ is a vertex cover of the induce subgraph $(L,\mathcal{S}(\mathbb{P})^L)$.}
\begin{proof}
($\Rightarrow$)
Suppose $L'$ is not a vertex cover of the graph $(L,\mathcal{S}(\mathbb{P})^L)$. Then, by items~\ref{dfitem:vertex cover cover} of Definition~\ref{df:vertex cover}, one of the next cases must be true:

Case 1: $L'\nsubseteq L$. Then, $L'$ is not a useful reduction of $L$ by item~\ref{dfitem:useful reduction} of Definition~\ref{df:minimum useful reduction}.

Case 2: $L'\subseteq L$ and there is an edge $e\in\mathcal{S}(\mathbb{P})^L$ such that $L'\cap e=\varnothing$.
Then, by Definition~\ref{df:vertex set reduced graph} and statement~\eqref{eq:game prohibition to graph edges}, there is a profile $\delta\in\mathbb{P}$ such that $e=L\cap\mathcal{S}(\delta)$.
Thus, $L'\cap(L\cap\mathcal{S}(\delta))=\varnothing$ by the case assumption $L'\cap e=\varnothing$.
Then, $L'\cap\mathcal{S}(\delta)=L'\cap L\cap\mathcal{S}(\delta)=\varnothing$ by the case assumption $L'\subseteq L$.
Thus, $L'$ is not a useful law in the game $(\mathcal{A},\Delta,\mathbb{P})$ by Lemma~\ref{lm:useful law}.
Hence, $L'$ is not a useful reduction of $L$ by item~\ref{dfitem:useful reduction} of Definition~\ref{df:minimum useful reduction}.

($\Leftarrow$)
Suppose $L'$ is not a useful reduction of $L$. Then, by item~\ref{dfitem:useful reduction} of  Definition~\ref{df:minimum useful reduction}, one of the next cases must be true: 

Case 1: $L'\nsubseteq L$. Then, $L'$ is not a vertex cover of the graph $(L,\mathcal{S}(\mathbb{P})^L)$ by item~\ref{dfitem:vertex cover cover} of Definition~\ref{df:vertex cover}.

Case 2: $L'$ is not useful. Then, by Lemma~\ref{lm:useful law}, there is a profile $\delta\in\mathbb{P}$ such that $L'\cap\mathcal{S}(\delta)=\varnothing$.
Thus,
\begin{equation}\label{eq:7-9-6}
    L'\cap(L\cap\mathcal{S}(\delta))=L\cap(L'\cap\mathcal{S}(\delta))=L\cap\varnothing=\varnothing.
\end{equation}
Note that $L\cap\mathcal{S}(\delta)\in\mathcal{S}(\mathbb{P})^L$ by statement~\eqref{eq:game prohibition to graph edges} and Definition~\ref{df:vertex set reduced graph}.
Hence, $L'$ is not a vertex cover of the graph $(L,\mathcal{S}(\mathbb{P})^L)$ by statement~\eqref{eq:7-9-6} and item~\ref{dfitem:vertex cover cover} of Definition~\ref{df:vertex cover}.
\end{proof}

\subsection{Algorithms for the $\mathbf{UL}$ Problems}\label{app_sec:algorithms for the UL problems}

\begin{algorithm}[bth]
\caption{Algorithms for the $\mathbf{UL}$ Problems}
\label{alg:UL design}
\SetKwFunction{FMain}{$\mathtt{IsUL}$}
\SetKwProg{Fn}{}{:}{}
\Fn{\FMain{game $(\mathcal{A},\Delta,\mathbb{P})$, set $L$}}{
    \Return $\mathtt{IsVC}$ (graph $(\bigcup\Delta,\mathcal{S}(\mathbb{P}))$, set $L$) \;
}

\SetKwFunction{FMain}{$\mathtt{IsMiniUL}$}
\SetKwProg{Fn}{}{:}{}
\Fn{\FMain{game $(\mathcal{A},\Delta,\mathbb{P})$, set $L$}}{
    \Return $\mathtt{IsMiniVC}$ (graph $(\bigcup\Delta,\mathcal{S}(\mathbb{P}))$, set $L$) \;
}

\SetKwFunction{FMain}{$\mathtt{AppMinUR}$}
\SetKwProg{Fn}{}{:}{}
\Fn{\FMain{game $(\mathcal{A},\Delta,\mathbb{P})$, set $L$}}{
    \If{$\mathtt{IsUL}$($(\mathcal{A},\Delta,\mathbb{P})$, $L$)}{
        \Return $\mathtt{AppMinVC}$ (graph $(L,\mathcal{S}(\mathbb{P})^L)$) \;
    }
    \Return \textit{false}\;
}
\end{algorithm}

\subsection{Approximation Factor of $\mathtt{AppMinUR}$}\label{app_sec:MinUR approximation factor}

We first prove the next lemma.

\begin{lemma}\label{lm:game and useful law to reduced graph}
For a useful law $L$ in a game $(\mathcal{A},\Delta,\mathbb{P})$, the tuple $(L,\mathcal{S}(\mathbb{P})^{L})$ is an $|\mathcal{A}|$-graph.
\end{lemma}
\begin{proof}
Note that $(\bigcup\Delta,\mathcal{S}(\mathbb{P}))$ is a $|\mathcal{A}|$-graph by Lemma~\ref{lm:game to graph} and $L$ is a vertex cover in the graph $(\bigcup\Delta,\mathcal{S}(\mathbb{P}))$ by Theorem~\ref{th:useful=vertex cover} and the assumption of the lemma that $L$ is a useful law in the game $(\mathcal{A},\Delta,\mathbb{P})$.
Then, the statement of the lemma follows from Definition~\ref{df:vertex set reduced graph}.
\end{proof}

Note that, by our assumption, algorithm $\mathtt{AppMinVC}$ in $k$-graphs has an approximation factor of $k$.
Meanwhile, by Lemma~\ref{lm:game and useful law to reduced graph}, the subgraph $(L,\mathcal{S}(\mathbb{P})^L)$ is an $|\mathcal{A}|$-graph.
Then, the vertex cover returned in line~7 of Algorithm~\ref{alg:UL design} (say $C$) is at most $|\mathcal{A}|$ times the minimum vertex cover of the input graph $(L,\mathcal{S}(\mathbb{P})^L)$ (say $\tilde{C}$):
\begin{equation}\label{eq:7-29-4}
    |C|\leq |\mathcal{A}|\times |\tilde{C}|.
\end{equation}
Meanwhile, by Theorem~\ref{th:useful reduction=vertex cover}, set $C$ is a useful reduction of $L$ and set $\tilde{C}$ is the minimum useful reduction of $L$.
Thus, statement~\eqref{eq:7-29-4} can be interpreted as: the size of the returned set in line~7 of Algorithm~\ref{alg:UL design} is at most $|\mathcal{A}|$ times the size of the minimum useful reduction.
Hence, algorithm $\mathtt{AppMinUR}$ in Algorithm~\ref{alg:UL design} is an $|\mathcal{A}|$-approximation of the problem $\mathbf{MinUR}$.

\section{Supplementary to Section~\ref{sec:gap-free law}}\label{app_sec:gap-free law}

\subsection{Proof of Theorem~\ref{th:reduce useful to gap-free}} \label{app_sec:proof of Theorem reduce useful to gap-free}

Definition~\ref{df:useful game to gap free game} indeed characterises a \textit{polynomial algorithm} to construct the game $(\bar{\mathcal{A}},\bar{\Delta},\bar{\mathbb{P}})$ for any given game $(\mathcal{A},\Delta,\mathbb{P})$.
To illustrate this definition, let us use the game in Figure~\ref{fig:ND intro game matrix} as the \textit{input} game $(\mathcal{A},\Delta,\mathbb{P})$, and construct the \textit{output} game $(\bar{\mathcal{A}},\bar{\Delta},\bar{\mathbb{P}})$, as illustrated in Figure~\ref{fig:Reduction_UL2GFL_game}. 

\begin{figure*}[htb]
    \centering
    \scalebox{0.5}{\includegraphics{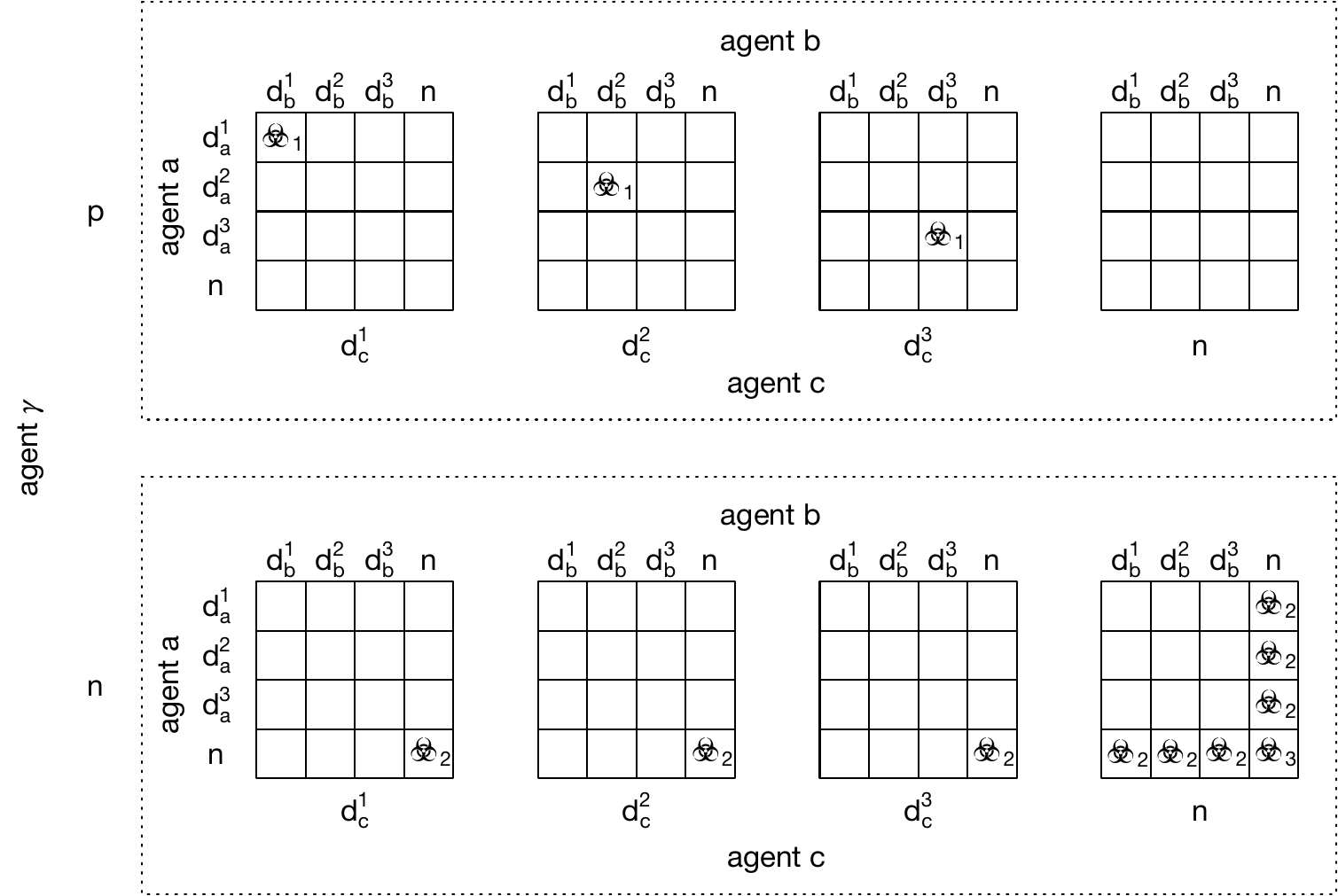}}
    \caption{An illustration of Definition~\ref{df:useful game to gap free game} using the game in Figure~\ref{fig:ND intro game matrix}.}
    \label{fig:Reduction_UL2GFL_game}
\end{figure*}

In this output game, an additional agent $\gamma$ is introduced alongside the original three agents $a,b,c$ (item~\ref{dfitem:useful game to gap free game agent} of Definition~\ref{df:useful game to gap free game}).
Each of the original agents $a, b, c$ has one more available action $n$, in addition to their original actions, while the newly introduced agent $\gamma$ has two available actions $p$ and $n$ (item~\ref{dfitem:useful game to gap free game action} of Definition~\ref{df:useful game to gap free game}).
If agent $\gamma$ chooses action $p$ (see the above part of Figure~\ref{fig:Reduction_UL2GFL_game}), then outcomes are prohibited only when the behaviours of the original agents constitute a prohibited outcome in the input game, as labelled by \Biohazard$_1$ in Figure~\ref{fig:Reduction_UL2GFL_game} (the first bullet in item~\ref{dfitem:useful game to gap free game prohibition} of Definition~\ref{df:useful game to gap free game}).
If agent $\gamma$ chooses action $n$ (see the below part of Figure~\ref{fig:Reduction_UL2GFL_game}), then outcomes are prohibited when exactly one of the original agents chooses their original actions, as labelled by \Biohazard$_2$ in Figure~\ref{fig:Reduction_UL2GFL_game} (the second bullet in item~\ref{dfitem:useful game to gap free game prohibition} of Definition~\ref{df:useful game to gap free game}), or when every original agent chooses action $n$, as labelled by \Biohazard$_3$ in Figure~\ref{fig:Reduction_UL2GFL_game} (the third bullet in item~\ref{dfitem:useful game to gap free game prohibition} of Definition~\ref{df:useful game to gap free game}).

To support the proof of Theorem~\ref{th:reduce useful to gap-free}, we first show the next lemma, which directly follows from Definition~\ref{df:useful game to gap free game}.

\begin{lemma}\label{lm:useful game to gap-free game profile support}
For any profile $\bar{\delta}\in\bar{\mathbb{P}}$ in a game $(\bar{\mathcal{A}},\bar{\Delta},\bar{\mathbb{P}})$,
\begin{enumerate}
    \item if $\bar{\delta}\in\bar{\mathbb{P}}_1$, then $\mathcal{S}(\bar{\delta})=\mathcal{S}(\delta)\cup\{p\}$ for a profile $\delta\in\mathbb{P}$;\label{lmitem:useful game to gap-free game profile support 1}
    \item if $\bar{\delta}\in\bar{\mathbb{P}}_2$, then $\mathcal{S}(\bar{\delta})=\{d,n\}$ for an action $d\in\bigcup\Delta$;\label{lmitem:useful game to gap-free game profile support 2}
    \item if $\bar{\delta}\in\bar{\mathbb{P}}_3$, then $\mathcal{S}(\bar{\delta})=\{n\}$.\label{lmitem:useful game to gap-free game profile support 3}
\end{enumerate}
\end{lemma}

Recall Lemma~\ref{lm:useful law} that a useful law $L$ in the game $(\mathcal{A},\Delta,\mathbb{P})$ intersects with $\mathcal{S}(\delta)$ for each profile $\delta\in\mathbb{P}$.
Then, by item~\ref{lmitem:useful game to gap-free game profile support 1} of the above lemma, the law $L$ intersects with $\mathcal{S}(\bar{\delta})$ for each profile $\bar{\delta}\in\bar{\mathbb{P}}_1$.
Also observe that $p,n\notin L$ and agent $\gamma$ takes action $n$ in each profile in the set $\bar{\mathbb{P}}_2\cup\bar{\mathbb{P}}_3$.
Then, by Lemma~\ref{lm:safe action in law imposed game}, the law $L$ makes $p$ a safe action of agent $\gamma$ in the game $(\bar{\mathcal{A}},\bar{\Delta}^L,\bar{\mathbb{P}}^L)$, which means that the law $L$ is gap-free in the game $(\bar{\mathcal{A}},\bar{\Delta},\bar{\mathbb{P}})$ by Lemma~\ref{lm:gap free law}.
Theorem~\ref{th:reduce useful to gap-free} below formalises the above observation and establishes its converse as well.

\noindent\textbf{Theorem~\ref{th:reduce useful to gap-free}. }
\textit{A set $L\subseteq\bigcup\Delta$ is a useful law in a game $(\mathcal{A},\Delta,\mathbb{P})$ if and only if $L$ is a gap-free law in the game $(\bar{\mathcal{A}},\bar{\Delta},\bar{\mathbb{P}})$.}
\begin{proof}
Note that $L\subseteq\bigcup\Delta\subsetneq\bigcup\bar{\Delta}$ by Definition~\ref{df:useful game to gap free game} and the assumption $L\subseteq\bigcup\Delta$ of the theorem.
Then, by Definition~\ref{df:law in game},
\begin{equation}\label{eq:7-4-3}
    \text{$L$ is a law in both games } (\mathcal{A},\Delta,\mathbb{P})\text{ and }(\bar{\mathcal{A}},\bar{\Delta},\bar{\mathbb{P}}),
\end{equation}
by Definition~\ref{df:useful game to gap free game},
\begin{equation}\label{eq:7-4-4}
    p,n\notin L,
\end{equation}
and, by item~\ref{dfitem:law-imposed game 1} of Definition~\ref{df:law-imposed game} and item~\ref{dfitem:useful game to gap free game action} of Definition~\ref{df:useful game to gap free game},
\begin{equation}\label{eq:7-4-2}
    \bar{\Delta}_a^L=\begin{cases}
            \Delta_a^L\cup\{n\}, &\!\text{if } a\in\mathcal{A};\\
            \{p, n\}, &\!\text{if } a=\gamma.
        \end{cases}
\end{equation}

($\Rightarrow$)
Suppose law $L$ is useful in the game $(\mathcal{A},\Delta,\mathbb{P})$. Then, 
\begin{equation}\label{eq:7-4-1}
    L\cap\mathcal{S}(\delta)\neq\varnothing \text{ for each profile } \delta\in\mathbb{P}
\end{equation}
by Lemma~\ref{lm:useful law}.
Consider an arbitrary profile $\bar{\delta}\in\bar{\mathbb{P}}$ such that
\begin{equation}\label{eq:7-6-1}
    \bar{\delta}_{\gamma}=p.
\end{equation}
Then, $\bar{\delta}\in\bar{\mathbb{P}}_1$ by item~\ref{dfitem:useful game to gap free game prohibition} of Definition~\ref{df:useful game to gap free game}.
Thus, $\mathcal{S}(\bar{\delta})=\mathcal{S}(\delta)\cup\{p\}$ for a profile $\delta\in\mathbb{P}$ by item~\ref{lmitem:useful game to gap-free game profile support 1} of Lemma~\ref{lm:useful game to gap-free game profile support}.
Then, $\mathcal{S}(\bar{\delta})\cap L\neq\varnothing$ by statements~\eqref{eq:7-4-1}.
Hence, $p$ is a safe action of agent $\gamma$ in the law-imposed game $(\bar{\mathcal{A}},\bar{\Delta}^L,\bar{\mathbb{P}}^L)$ by statements~\eqref{eq:7-4-2}, \eqref{eq:7-6-1}, Lemma~\ref{lm:safe action in law imposed game}, and the arbitrariness of $\bar{\delta}$.
Therefore, law $L$ is gap-free in the game $(\bar{\mathcal{A}},\bar{\Delta},\bar{\mathbb{P}})$ by Lemma~\ref{lm:gap free law} and statement~\eqref{eq:7-4-3}.

($\Leftarrow$)
Suppose $L$ is not useful in the game $(\mathcal{A},\Delta,\mathbb{P})$. Then, there is a profile $\delta\in\mathbb{P}$ such that $L\cap\mathcal{S}(\delta)=\varnothing$ by Lemma~\ref{lm:useful law}.
Thus, there is a profile $\bar{\delta}\in\bar{\mathbb{P}}_1$ such that $\bar{\delta}_{\gamma}=p$ and $L\cap\mathcal{S}(\bar{\delta})=\varnothing$ by item~\ref{dfitem:useful game to gap free game prohibition} of Definition~\ref{df:useful game to gap free game}, item~\ref{lmitem:useful game to gap-free game profile support 1} of Lemma~\ref{lm:useful game to gap-free game profile support}, and statement~\eqref{eq:7-4-4}.
Hence, by Lemma~\ref{lm:safe action in law imposed game},
\begin{equation}\label{eq:7-4-5}
    \text{$p$ is not a safe action of agent $\gamma$}
\end{equation}
in the law-imposed game $(\bar{\mathcal{A}},\bar{\Delta}^L,\bar{\mathbb{P}}^L)$.

Consider an arbitrary agent $a\in\mathcal{A}$ and an arbitrary action $d\in\Delta_a^L$.
Then, $d\notin L$ by item~\ref{dfitem:law-imposed game 1} of Definition~\ref{df:law-imposed game} and there is a profile $\bar{\delta}'\in\bar{\mathbb{P}}_2$ such that $\bar{\delta}'_a=d$ and $\mathcal{S}(\bar{\delta}')=\{d,n\}$ by item~\ref{dfitem:useful game to gap free game prohibition} of Definition~\ref{df:useful game to gap free game} and item~\ref{lmitem:useful game to gap-free game profile support 2} of Lemma~\ref{lm:useful game to gap-free game profile support}.
Thus, $\mathcal{S}(\bar{\delta}')\cap L=\varnothing$ by statement~\eqref{eq:7-4-4}.
Hence, by Lemma~\ref{lm:safe action in law imposed game}, for any agent $a\in\mathcal{A}$ and any action $d\in\Delta_a^L$,
\begin{equation}\label{eq:7-4-6}
    \text{$d$ is not a safe action of agent $a$}
\end{equation}
in the law-imposed game $(\bar{\mathcal{A}},\bar{\Delta}^L,\bar{\mathbb{P}}^L)$.

Moreover, there is a profile $\bar{\delta}''\in\bar{\mathbb{P}}_3$ such that $\bar{\delta}''_a=n$ for each agent $a\in\bar{\mathcal{A}}$ and $\mathcal{S}(\bar{\delta}'')=\{n\}$ by item~\ref{dfitem:useful game to gap free game prohibition} of Definition~\ref{df:useful game to gap free game} and item~\ref{lmitem:useful game to gap-free game profile support 3} of Lemma~\ref{lm:useful game to gap-free game profile support}.
Thus, $\mathcal{S}(\bar{\delta}'')\cap L=\varnothing$ by statement~\eqref{eq:7-4-4}.
Hence, by Lemma~\ref{lm:useful law} and that $\bar{\delta}''\in\bar{\mathbb{P}}_3\subseteq \bar{\mathbb{P}}$,
\begin{equation}\label{eq:7-4-8}
    \text{$L$ is not a useful law in the game $(\bar{\mathcal{A}},\bar{\Delta},\bar{\mathbb{P}})$}
\end{equation}
and, by Lemma~\ref{lm:safe action in law imposed game}, for any agent $a\in\bar{\mathcal{A}}$,
\begin{equation}\label{eq:7-4-7}
    \text{$n$ is not a safe action of agent $a$}
\end{equation}
in the law-imposed game $(\bar{\mathcal{A}},\bar{\Delta}^L,\bar{\mathbb{P}}^L)$.

In conclusion, by statements~\eqref{eq:7-4-2}, \eqref{eq:7-4-5}, \eqref{eq:7-4-6}, and \eqref{eq:7-4-7}, for every agent $a\in\bar{\mathcal{A}}$, every action $d\in\bar{\Delta}_a^L$ is not a safe action of agent $a$ in the law-imposed game $(\bar{\mathcal{A}},\bar{\Delta}^L,\bar{\mathbb{P}}^L)$.
Therefore, law $L$ is not gap-free in the game $(\bar{\mathcal{A}},\bar{\Delta},\bar{\mathbb{P}})$ by Lemma~\ref{lm:gap free law} and statement~\eqref{eq:7-4-8}.
\end{proof}

\subsection{Proof of Theorem~\ref{th:MinGFR hard to approximate}}\label{app_sec:proof of Theorem MinGFR hard to approximate}

\noindent\textbf{Theorem~\ref{th:MinGFR hard to approximate}. }
\textit{$\mathbf{MinGFR}$ in a game $(\mathcal{A},\Delta,\mathbb{P})$ is NP-hard to approximate within factor $|\mathcal{A}|-1-\epsilon$ for any $\epsilon\!>\!0$ when $|\mathcal{A}|\!\geq\! 3$.}
\begin{proof}
Suppose the opposite.
Then, there is an $\epsilon>0$ and an algorithm $\mathtt{Alg}$ that approximates $\mathbf{MinGFR}$ with factor $|\mathcal{A}|-1-\epsilon$ for any game $(\mathcal{A},\Delta,\mathbb{P})$ where $|\mathcal{A}|\geq 3$.

Consider an arbitrary useful law $L$ in an arbitrary game $(\mathcal{A},\Delta,\mathbb{P})$ where $|\mathcal{A}|\geq 2$. We consider the next two steps:
\begin{enumerate}
    \item compute the game $(\bar{\mathcal{A}},\bar{\Delta},\bar{\mathbb{P}})$ corresponding to the game $(\mathcal{A},\Delta,\mathbb{P})$ by Definition~\ref{df:useful game to gap free game};
    \item use the algorithm $\mathtt{Alg}$ to get an approximated gap-free reduction $L'$ of law $L$ in the game $(\bar{\mathcal{A}},\bar{\Delta},\bar{\mathbb{P}})$.
\end{enumerate}

Note that, by the assumption $|\mathcal{A}|\geq 2$ and Definition~\ref{dfitem:useful game to gap free game agent} of Definition~\ref{df:useful game to gap free game},
\begin{equation}\label{eq:7-28-1}
    |\bar{\mathcal{A}}|=|\mathcal{A}|+1\geq 3.
\end{equation}

By step~2 above, the assumption that $L$ is a useful law in the game $(\mathcal{A},\Delta,\mathbb{P})$, and Corollary~\ref{cr:useful reduction to gap-free reduction},
\begin{equation}\label{eq:7-30-2}
    \text{$L'$ is a useful reduction of $L$ in the game $(\mathcal{A},\Delta,\mathbb{P})$.}
\end{equation}

On the other hand, step~2 above implies that $|L'|$ is at most $(|\bar{\mathcal{A}}|-1-\epsilon)$ times the size of the minimum gap-free reduction of law $L$ in the game $(\bar{\mathcal{A}},\bar{\Delta},\bar{\mathbb{P}})$ by the assumption about the approximation factor of $\mathtt{Alg}$ and statement~\eqref{eq:7-28-1}.
Meanwhile, the minimum gap-free reduction of law $L$ in the game $(\bar{\mathcal{A}},\bar{\Delta},\bar{\mathbb{P}})$ is the minimum useful reduction of law $L$ in the game $(\mathcal{A},\Delta,\mathbb{P})$ by Corollary~\ref{cr:useful reduction to gap-free reduction}.
Then, by statement~\eqref{eq:7-28-1}, the former two facts imply that $|L'|$ is at most $(|\mathcal{A}|-\epsilon)$ times the size of the minimum useful reduction of $L$ in the game $(\mathcal{A},\Delta,\mathbb{P})$.
Hence, by statement~\eqref{eq:7-30-2} and the arbitrariness of the useful law $L$ and the game $(\mathcal{A},\Delta,\mathbb{P})$, the above two steps form an algorithm that approximates $\mathbf{MinUR}$ with factor $|\mathcal{A}|-\epsilon$, which contradicts Theorem~\ref{th:MinUR hard to approximate}.
\end{proof}

\subsection{Proof of Lemma~\ref{lm:safable action}}\label{app_sec:proof of Lemma safable action}

\noindent\textbf{Lemma~\ref{lm:safable action}. }
\textit{For a game $(\mathcal{A},\Delta,\mathbb{P})$, an action $d\in\bigcup\Delta$ is safable if and only if $\mathcal{S}(\delta)\neq\{d\}$ for each profile $\delta\in\mathbb{P}$.}
\begin{proof}
($\Rightarrow$)
Suppose there is a profile $\delta$ such that
\begin{equation}\label{eq:7-9-3}
    \delta\in\mathbb{P} \text{ and } \mathcal{S}(\delta)=\{d\}.
\end{equation}
Then, by statement~\eqref{eq:profile support set},
\begin{equation}\label{eq:7-9-4}
    \delta_a=d \text{ for each agent } a\in\mathcal{A}.
\end{equation}

Consider an arbitrary agent $a\in\mathcal{A}$ and an arbitrary law $L\subseteq\bigcup\Delta$. 
Then, one of the next cases must be true:

Case 1: $d\in L$. Then, $d\notin\Delta_a^L$ by item~\ref{dfitem:law-imposed game 1} of Definition~\ref{df:law-imposed game}.
Thus, $d$ is not a safe action of agent $a$ in the law-imposed game $(\mathcal{A},\Delta^L,\mathbb{P}^L)$ by Lemma~\ref{lm:safe action in law imposed game}.

Case 2: $d\notin L$. Then, $L\cap\mathcal{S}(\delta)=\varnothing$ by the right part of statement~\eqref{eq:7-9-3}.
Thus, $d$ is not a safe action of agent $a$ in the law-imposed game $(\mathcal{A},\Delta^L,\mathbb{P}^L)$ by the left part of statement~\eqref{eq:7-9-3}, statement~\eqref{eq:7-9-4}, and Lemma~\ref{lm:safe action in law imposed game}.

In conclusion, action $d$ is not safable by Definition~\ref{df:safable action} and the arbitrariness of $a$ and $L$.

($\Leftarrow$)
Suppose $\mathcal{S}(\delta)\neq\{d\}$ for each profile $\delta\in\mathbb{P}$. 
Let $L$ be the law such that 
\begin{equation}\textstyle\label{eq:7-2-1}
    L=\bigcup_{\delta\in\mathbb{P},\delta_a=d}(\mathcal{S}(\delta)\setminus\{d\}).
\end{equation}
Then,
\begin{equation}\label{eq:7-25-1}
    d\notin L.
\end{equation}
Note that, by the assumption $d\in\bigcup\Delta$ of the lemma and item~\ref{dfitem:game actions} of Definition~\ref{df:game}, there is an agent $a\in\mathcal{A}$ such that $d\in\Delta_a$.
Thus, by statement~\eqref{eq:7-25-1} and item~\ref{dfitem:law-imposed game 1} of Definition~\ref{df:law-imposed game},
\begin{equation}\label{eq:7-9-1}
    d\in\Delta_a^L.
\end{equation}

If there is no profile $\delta\in\mathbb{P}$ such that $\delta_a=d$, then it is vacuously true that $d$ is a safe action of agent $a$ in the game $(\mathcal{A},\Delta^L,\mathbb{P}^L)$ by statement~\eqref{eq:7-9-1} and Lemma~\ref{lm:safe action in law imposed game}.

Otherwise, consider an arbitrary profile $\delta'\in\mathbb{P}$ such that $\delta'_a=d$.
Then, 
\begin{equation}\label{eq:7-9-2}
     (\mathcal{S}(\delta')\setminus\{d\})\cap L=\mathcal{S}(\delta')\setminus\{d\}
\end{equation}
by statement~\eqref{eq:7-2-1}, and $\mathcal{S}(\delta')\neq\{d\}$ by the ($\Leftarrow$) part assumption.
Thus, $\mathcal{S}(\delta')\cap L\supseteq (\mathcal{S}(\delta')\setminus\{d\})\cap L\neq\varnothing$ by statement~\eqref{eq:7-9-2}.
Hence, $d$ is a safe action of agent $a$ in the law-imposed game $(\mathcal{A},\Delta^L,\mathbb{P}^L)$ by statement~\eqref{eq:7-9-1}, Lemma~\ref{lm:safe action in law imposed game}, and the arbitrariness of $\delta'$.

In conclusion, action $d$ is safable by Definition~\ref{df:safable action}.
\end{proof}

\subsection{Proof of Lemma~\ref{lm:safe action=vertex cover in reduced graph}}\label{app_sec:proof of Lemma safe action=vertex cover in reduced graph}

To illustrate Definition~\ref{df:safe action verify graph}, let us use the game in Figure~\ref{fig:Reduction_UL2GFL_game}.
Consider the action $p$ of agent $\gamma$, where $p$ is safable by Lemma~\ref{lm:safable action}.
The induced graph by Definition~\ref{df:safe action verify graph} is visualised in Figure~\ref{fig:Reduction_GFL2VC_graph}.
In particular, by item~\ref{dfitem:safe action verify graph vertex} of Definition~\ref{df:safe action verify graph}, the graph consists of ten vertices corresponding to the ten (out of eleven, except for action $p$) actions in Figure~\ref{fig:Reduction_UL2GFL_game}.
Meanwhile, there are three prohibited outcomes $\delta$ such that $\delta_\gamma=p$ in Figure~\ref{fig:Reduction_UL2GFL_game}, as labelled by the sign \Biohazard$_1$. 
Accordingly, by item~\ref{dfitem:safe action verify graph edge} of Definition~\ref{df:safe action verify graph}, there are three edges in Figure~\ref{fig:Reduction_GFL2VC_graph}, each corresponding to a \Biohazard$_1$-labelled outcome in Figure~\ref{fig:Reduction_UL2GFL_game}. For instance, edge $e_1=\{d_a^1,d_b^1,d_c^1\}$ in Figure~\ref{fig:Reduction_GFL2VC_graph} corresponds to the left-most \Biohazard$_1$-labelled outcome in Figure~\ref{fig:Reduction_UL2GFL_game}, which consists of the actions $p,d_a^1,d_b^1$, and $d_c^1$. The same is true for edge $e_2$ in  Figure~\ref{fig:Reduction_GFL2VC_graph} and the middle \Biohazard$_1$-labelled outcome in Figure~\ref{fig:Reduction_UL2GFL_game}, as well as edge $e_3$ in  Figure~\ref{fig:Reduction_GFL2VC_graph} and the right-most \Biohazard$_1$-labelled outcome in Figure~\ref{fig:Reduction_UL2GFL_game}.

\begin{figure}[htb]
    \centering
    \scalebox{0.5}{\includegraphics{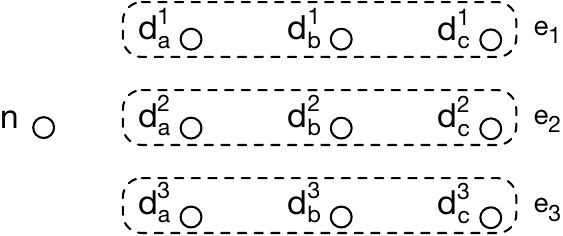}}
    \caption{An illustration of Definition~\ref{df:safe action verify graph} using the the game in Figure~\ref{fig:Reduction_UL2GFL_game}, agent $\gamma$, and action $p$.}
    \label{fig:Reduction_GFL2VC_graph}
\end{figure}

Observe that Definition~\ref{df:safe action verify graph} indeed implies a \textit{polynomial algorithm} to convert a game $(\mathcal{A},\Delta,\mathbb{P})$, an agent $a\in\mathcal{A}$, and an action $d\in\Delta_a$ into the graph $\mathcal{H}^{a,d}_{(\mathcal{A},\Delta,\mathbb{P})}$.
The next lemma formally captures the aforementioned observation about their relation and thereby bridges ``a safe action in a law-imposed game'' with a vertex cover in the induced game.

\noindent\textbf{Lemma~\ref{lm:safe action=vertex cover in reduced graph}. }
\textit{For a law $L$ in a game $(\mathcal{A},\Delta,\mathbb{P})$ and an agent $a\in\mathcal{A}$, an action $d\in\bigcup\Delta$ is a safe action of agent $a$ in the law-imposed game $(\mathcal{A},\Delta^L,\mathbb{P}^L)$ if and only if $d\in\Delta_a^L$, $d$ is safable in the game $(\mathcal{A},\Delta,\mathbb{P})$, and $L$ is a vertex cover in the graph $\mathcal{H}^{a,d}_{(\mathcal{A},\Delta,\mathbb{P})}$.}
\begin{proof}
($\Rightarrow$)
Suppose $d$ is a safe action of agent $a$ in the law-imposed game $(\mathcal{A},\Delta^L,\mathbb{P}^L)$.
Then, $d$ is safable in the game $(\mathcal{A},\Delta,\mathbb{P})$ by Definition~\ref{df:safable action} and $d\in\Delta_a^L$ by Definition~\ref{df:safe action}.
Thus, by item~\ref{dfitem:law-imposed game 1} of Definition~\ref{df:law-imposed game},
\begin{equation}\label{eq:7-2-5}
    d\in\Delta_a \text{ and } d\notin L.
\end{equation}

Toward proving that $L$ is a vertex cover in the graph $\mathcal{H}^{a,d}_{(\mathcal{A},\Delta,\mathbb{P})}$,
consider an arbitrary profile $\delta\in\mathbb{P}$ such that $\delta_a=d$.
Then, $L\cap\mathcal{S}(\delta)\neq\varnothing$ by Lemma~\ref{lm:safe action in law imposed game} and the ($\Rightarrow$) part assumption.
Thus, $L\cap(\mathcal{S}(\delta)\setminus\{d\})\neq\varnothing$ by the right part of statement~\eqref{eq:7-2-5}.
Hence, by item~\ref{dfitem:safe action verify graph edge} of Definition~\ref{df:safe action verify graph} and the arbitrariness of $\delta$,
\begin{equation}\label{eq:7-9-5}
    L\cap e\neq\varnothing \text{ for each edge $e\in\mathcal{E}^{a,d}_{(\mathcal{A},\Delta,\mathbb{P})}$.}
\end{equation}
Meanwhile, $L\subseteq\bigcup\Delta$ by Definition~\ref{df:law in game} and the assumption that $L$ is a law of the game $(\mathcal{A},\Delta,\mathbb{P})$ of the lemma.
Then, $L\subseteq\mathcal{V}^{a,d}_{(\mathcal{A},\Delta,\mathbb{P})}$ by the right part of statement~\eqref{eq:7-2-5} and item~\ref{dfitem:safe action verify graph vertex} of Definition~\ref{df:safe action verify graph}.
Hence, $L$ is a vertex cover of the graph $\mathcal{H}^{a,d}_{(\mathcal{A},\Delta,\mathbb{P})}$ by Definition~\ref{df:safe action verify graph}, statement~\eqref{eq:7-9-5}, and item~\ref{dfitem:vertex cover cover} of Definition~\ref{df:vertex cover}.

$(\Leftarrow)$
Suppose $d\in\Delta_a^L$ and $L$ is a vertex cover in the graph $\mathcal{H}^{a,d}_{(\mathcal{A},\Delta,\mathbb{P})}$.
Then, $L\cap(\mathcal{S}(\delta)\setminus\{d\})\neq\varnothing$ for each profile $\delta\in\mathbb{P}$ such that $\delta_a=d$ by item~\ref{dfitem:safe action verify graph edge} of Definition~\ref{df:safe action verify graph} and item~\ref{dfitem:vertex cover cover} of Definition~\ref{df:vertex cover}.
Thus, $L\cap\mathcal{S}(\delta)\neq\varnothing$ for each profile $\delta\in\mathbb{P}$ such that $\delta_a=d$.
Hence, $d$ is a safe action of agent $a$ in the law-imposed game $(\mathcal{A},\Delta^L,\mathbb{P}^L)$ by Lemma~\ref{lm:safe action in law imposed game} and the $(\Leftarrow)$ part assumption $d\in\Delta_a^L$.
\end{proof}

By the above lemma, for a law $L$ in the game in Figure~\ref{fig:Reduction_UL2GFL_game}, action $p$ is a safe action of agent $\lambda$ under law $L$ if and only if $L$ is a vertex cover in the graph in Figure~\ref{fig:Reduction_GFL2VC_graph}.
This property, however, is easily observable through another chain of equivalence.
Observe that a vertex cover in the graph in Figure~\ref{fig:Reduction_GFL2VC_graph} is exactly a vertex cover in the graph in Figure~\ref{fig:Reduction_UL2VC_originalgraph}, the latter of which, by Theorem~\ref{th:useful=vertex cover}, is a useful law in the game in Figure~\ref{fig:ND intro game matrix}.
Such a useful law, by Theorem~\ref{th:reduce useful to gap-free} and its proof, is gap-free in the game in Figure~\ref{fig:Reduction_UL2GFL_game} by rendering action $p$ a safe action of agent $\gamma$ in the law-imposed game. 
This observation provides insight into the construction presented in Definition~\ref{df:useful game to gap free game}.

\subsection{Proof of Theorem~\ref{th:reduce minimal gap-free to multiple minimal vertex cover}}\label{app_sec:proof Theorem reduce minimal gap-free to multiple minimal vertex cover}

The next lemma will be used in the later proofs. It follows directly from item~\ref{dfitem:vertex cover cover} of Definition~\ref{df:vertex cover}.
 
\begin{lemma}\label{lm:vertex cover monotonicity}
For a vertex cover $C$ in a graph $(V,E)$, every set $C'$ such that $C\subseteq C'\subseteq V$ is a vertex cover of the graph. 
\end{lemma}

To resolve the complexity in the proof of Theorem~\ref{th:reduce minimal gap-free to multiple minimal vertex cover}, we establish Lemma~\ref{lm:minimal gap free} below. 
It shows that, to verify whether a gap-free law $L$ is minimal, it suffices to check, for each action $d\in L$, whether removing $d$ from law $L$ preserves gap-freeness.

\begin{lemma}\label{lm:minimal gap free}
A gap-free law $L$ in a game $(\mathcal{A},\Delta,\mathbb{P})$ is minimal if and only if $L\setminus\{d\}$ is not gap-free for each action $d\in L$.
\end{lemma}
\begin{proof}
($\Rightarrow$)
This part of the statement follows from Definition~\ref{df:minimal gap free} because $L\setminus\{d\}\subsetneq L$ for each action $d\in L$.

($\Leftarrow$)
Suppose a gap-free law $L$ is not minimal.
Then, 
\begin{equation}\label{eq:7-8-3}
    \text{there is a gap-free law $L'\subsetneq L$} 
\end{equation}
by Definition~\ref{df:minimal gap free}. Thus,
\begin{equation}\label{eq:7-8-1}
    \text{there is an action $d_0\in L$ such that $L'\subseteq L\setminus\{d_0\}$}
\end{equation}
and by Lemma~\ref{lm:gap free law}, one of the next cases must be true:

Case 1: $L'$ is useful. Then, $L'\cap\mathcal{S}(\delta)\neq\varnothing$ for each profile $\delta\in\mathbb{P}$ by Lemma~\ref{lm:useful law}. 
Thus, $(L\setminus\{d_0\})\cap\mathcal{S}(\delta)\neq\varnothing$ for each profile $\delta\in\mathbb{P}$ by statement~\eqref{eq:7-8-1}.
Then, $L\setminus\{d_0\}$ is a useful law in the game $(\mathcal{A},\Delta,\mathbb{P})$ by Lemma~\ref{lm:useful law}. 
Hence, $L\setminus\{d_0\}$ is a gap-free law in the game $(\mathcal{A},\Delta,\mathbb{P})$ by Lemma~\ref{lm:gap free law}, where $d_0\in L$ by statement~\eqref{eq:7-8-1}.

Case 2: there is an agent $a\in\mathcal{A}$ and a safe action $d_1$ of agent $a$ in the law-imposed game $(\mathcal{A},\Delta^{L'},\mathbb{P}^{L'})$.
Then, 
\begin{equation}\label{eq:7-8-2}
    d_1\in\Delta_a^{L'} \text{ and } L'\cap\mathcal{S}(\delta)\neq\varnothing
\end{equation}
for each profile $\delta\in\mathbb{P}$ such that $\delta_a=d_1$ by Lemma~\ref{lm:safe action in law imposed game}.
Thus,
\begin{equation}\label{eq:7-8-5}
    d_1\in\Delta_a \text{ and } d_1\notin L'
\end{equation}
by  item~\ref{dfitem:law-imposed game 1} of Definition~\ref{df:law-imposed game}.

Subcase 2.1: $d_1\in L$. 
Then,
\begin{equation}\label{eq:7-8-4}
    d_1\in\Delta_a\setminus(L\setminus\{d_1\})=\Delta_a^{L\setminus\{d_1\}}
\end{equation}
and $L'\subseteq L\setminus\{d_1\}$ by statements~\eqref{eq:7-8-5}, \eqref{eq:7-8-3}, and item~\ref{dfitem:law-imposed game 1} of Definition~\ref{df:law-imposed game}.
Thus, $(L\setminus\{d_1\})\cap\mathcal{S}(\delta)\neq\varnothing$ for each profile $\delta\in\mathbb{P}$ such that $\delta_a=d_1$ by the right part of statement~\eqref{eq:7-8-2}.
Then, $d_1$ is a safe action of agent $a$ in the law-imposed game $(\mathcal{A},\Delta^{L\setminus\{d_1\}},\mathbb{P}^{L\setminus\{d_1\}})$ by Lemma~\ref{lm:safe action in law imposed game} and statement~\eqref{eq:7-8-4}.
Hence, law $L\setminus\{d_1\}$ is a gap-free by Lemma~\ref{lm:gap free law}, where $d_1\in L$ by the subcase assumption.

Subcase 2.2: $d_1\notin L$.
Then, $d_1\notin L\setminus\{d_0\}$.
Thus, by the left part of statement~\eqref{eq:7-8-5} and item~\ref{dfitem:law-imposed game 1} of Definition~\ref{df:law-imposed game},
\begin{equation}\label{eq:7-8-6}
    d_1\in\Delta_a\setminus(L\setminus\{d_0\})=\Delta_a^{L\setminus\{d_0\}}.
\end{equation}
Note that $(L\setminus\{d_0\})\cap\mathcal{S}(\delta)\neq\varnothing$ for each profile $\delta\in\mathbb{P}$ such that $\delta_a=d_1$ by statements~\eqref{eq:7-8-2} and \eqref{eq:7-8-1}.
Then, $d_1$ is a safe action of agent $a$ in the law-imposed game $(\mathcal{A},\Delta^{L\setminus\{d_0\}},\mathbb{P}^{L\setminus\{d_0\}})$ by Lemma~\ref{lm:safe action in law imposed game} and statement~\eqref{eq:7-8-6}.
Hence, law $L\setminus\{d_0\}$ is a gap-free by Lemma~\ref{lm:gap free law}, where $d_0\in L$ by statement~\eqref{eq:7-8-1}.
\end{proof}

\noindent\textbf{Theorem~\ref{th:reduce minimal gap-free to multiple minimal vertex cover}. }
\textit{A gap-free law $L$ in a game $(\mathcal{A},\Delta,\mathbb{P})$ is minimal if and only if all of the following statements are true:
\begin{enumerate}
    \item if $L$ is a vertex cover of the graph $(\bigcup\Delta,\mathcal{S}(\mathbb{P}))$, then $L$ is a minimal vertex cover in this graph;\label{ppitem:reduce minimal gap-free to multiple minimal vertex cover 1}
    \item for each agent $a\in\mathcal{A}$ and each safable action $d\in\Delta_a^L$, if $L$ is a vertex cover in graph $\mathcal{H}^{a,d}_{(\mathcal{A},\Delta,\mathbb{P})}$, then $L$ is a minimal vertex cover in this graph;\label{ppitem:reduce minimal gap-free to multiple minimal vertex cover 2}
    \item for each agent $a\in\mathcal{A}$ and each safable action $d\!\in\!\Delta_a\!\cap\! L$, the set $L\!\setminus\!\{d\}$ is not a vertex cover in graph $\mathcal{H}^{a,d}_{(\mathcal{A},\Delta,\mathbb{P})}$.\label{ppitem:reduce minimal gap-free to multiple minimal vertex cover 3}
\end{enumerate}}
\begin{proof}
($\Rightarrow$)
If item~\ref{ppitem:reduce minimal gap-free to multiple minimal vertex cover 1} is false,
then, by item~\ref{dfitem:vertex cover minimal} of Definition~\ref{df:vertex cover}, there is a set $L'$ such that 
\begin{equation}\label{eq:7-11-1}
    L'\subsetneq L
\end{equation}
being a vertex cover of the graph $(\bigcup\Delta,\mathcal{S}(\mathbb{P}))$.
Thus, $L'$ is a gap-free law in the game $(\mathcal{A},\Delta,\mathbb{P})$ by Theorem~\ref{th:reduce gap-free to multiple vertex cover}.
Hence, $L$ is not minimal-gap-free by Definition~\ref{df:minimal gap free} and statement~\eqref{eq:7-11-1}.

If item~\ref{ppitem:reduce minimal gap-free to multiple minimal vertex cover 2} is false,
then, by item~\ref{dfitem:vertex cover minimal} of Definition~\ref{df:vertex cover}, there is an agent $a$, a safable action $d\in\Delta_a^L$, and a vertex cover $L'$ in the graph $\mathcal{H}^{a,d}_{(\mathcal{A},\Delta,\mathbb{P})}$ such that
\begin{equation}\label{eq:7-11-2}
    L'\subsetneq L.
\end{equation}
Thus, $d\in\Delta_a^L=\Delta_a\setminus L\subseteq\Delta_a\setminus L'=\Delta_a^{L'}$ by item~\ref{dfitem:law-imposed game 1} of Definition~\ref{df:law-imposed game}.
Then, action $d$ is a safe action of agent $a$ in the law-imposed game $(\mathcal{A},\Delta^{L'},\mathbb{P}^{L'})$ by Lemma~\ref{lm:safe action=vertex cover in reduced graph}.
Hence, law $L'$ is gap-free in the game $(\mathcal{A},\Delta,\mathbb{P})$ by Lemma~\ref{lm:gap free law}.
Then, $L$ is not minimal-gap-free by statement~\eqref{eq:7-11-2} and Definition~\ref{df:minimal gap free}.

If item~\ref{ppitem:reduce minimal gap-free to multiple minimal vertex cover 3} is false,
then there is an agent $a$ and a safable action $d\in\Delta_a\cap L$ such that
\begin{equation}\label{eq:7-11-3}
    \text{$L\setminus\{d\}$ is a vertex cover in the graph $\mathcal{H}^{a,d}_{(\mathcal{A},\Delta,\mathbb{P})}$.}
\end{equation}
Thus, $d\in\Delta_a$ and $d\in L$.
Then,
\begin{equation}\label{eq:7-11-4}
    L\setminus\{d\}\subsetneq L
\end{equation}
and $d\in\Delta_a\setminus(L\setminus\{d\})=\Delta_a^{L\setminus\{d\}}$ by item~\ref{dfitem:law-imposed game 1} of Definition~\ref{df:law-imposed game}.
Hence, action $d$ is a safe action of agent $a$ in the law-imposed game $(\mathcal{A},\Delta^{L\setminus\{d\}},\mathbb{P}^{L\setminus\{d\}})$ by Lemma~\ref{lm:safe action=vertex cover in reduced graph}, statement~\eqref{eq:7-11-3}, and that $d$ is safable.
Then, law $L\setminus\{d\}$ is gap-free in the game $(\mathcal{A},\Delta,\mathbb{P})$ by Lemma~\ref{lm:gap free law}.
Therefore, $L$ is not minimal-gap-free by Definition~\ref{df:minimal gap free} and statement~\eqref{eq:7-11-4}.

($\Leftarrow$)
Suppose gap-free law $L$ is not minimal.
Then, by Lemma~\ref{lm:minimal gap free}, there is an action $d$ such that
\begin{equation}\label{eq:7-11-6}
    d \in L
\end{equation}
and $L\setminus\{d\}$ is gap-free.
Thus,
\begin{equation}\textstyle\label{eq:7-11-5}
    L\setminus\{d\}\subsetneq L\subseteq\bigcup\Delta
\end{equation}
and, by Theorem~\ref{th:reduce gap-free to multiple vertex cover}, one of the next cases must be true:

Case 1: $L\setminus\{d\}$ is a vertex cover of the graph $(\bigcup\Delta,\mathcal{S}(\mathbb{P}))$.
Then, $L$ is a vertex cover of the graph $(\bigcup\Delta,\mathcal{S}(\mathbb{P}))$ by Lemma~\ref{lm:vertex cover monotonicity} and statement~\eqref{eq:7-11-5}.
Meanwhile, $L$ is not a minimal vertex cover of the graph $(\bigcup\Delta,\mathcal{S}(\mathbb{P}))$ by the case assumption, statement~\eqref{eq:7-11-5}, and item~\ref{dfitem:vertex cover minimal} of Definition~\ref{df:vertex cover}.
Therefore, item~\ref{ppitem:reduce minimal gap-free to multiple minimal vertex cover 1} of the theorem is false.

Case 2: there is an agent $a\in\mathcal{A}$ and a safable action $d'\in\Delta_a^{L\setminus\{d\}}$ such that
\begin{equation}\label{eq:7-11-7}
    \text{$L\setminus\{d\}$ is a vertex cover of the graph $\mathcal{H}^{a,d'}_{(\mathcal{A},\Delta,\mathbb{P})}$.}
\end{equation}
Note that $\Delta_a^{L}\subseteq\Delta_a^{L\setminus\{d\}}$ by by item~\ref{dfitem:law-imposed game 1} of Definition~\ref{df:law-imposed game}.
Then, by the case assumption $d'\in\Delta_a^{L\setminus\{d\}}$, one of the next subcases must be true:

Subcase 2.1: $d'\in\Delta_a^{L}$.
Then, $d'\notin L$ by item~\ref{dfitem:law-imposed game 1} of Definition~\ref{df:law-imposed game}. 
Thus, $L\subseteq\bigcup\Delta\setminus\{d'\}=\mathcal{V}^{a,d'}_{(\mathcal{A},\Delta,\mathbb{P})}$ by statement~\eqref{eq:7-11-5} and item~\ref{dfitem:safe action verify graph vertex} of Definition~\ref{df:safe action verify graph}.
Hence, $L$ is a vertex cover of the graph $\mathcal{H}^{a,d'}_{(\mathcal{A},\Delta,\mathbb{P})}$ by statement~\eqref{eq:7-11-7} and Lemma~\ref{lm:vertex cover monotonicity}.
Meanwhile, $L$ is not a minimal vertex cover of the graph $\mathcal{H}^{a,d'}_{(\mathcal{A},\Delta,\mathbb{P})}$ by statements~\eqref{eq:7-11-7}, \eqref{eq:7-11-5}, and item~\ref{dfitem:vertex cover minimal} of Definition~\ref{df:vertex cover}.
Therefore, item~\ref{ppitem:reduce minimal gap-free to multiple minimal vertex cover 2} of the theorem is false by the subcase assumption $d'\in\Delta_a^L$ and the case assumption $d'$ is safable.

Subcase 2.2: $d'\in\Delta_a^{L\setminus\{d\}}\setminus\Delta_a^{L}$.
It is worth noting that $\Delta_a^{L\setminus\{d\}}\setminus\Delta_a^L=\Delta_a\cap\{d\}$ by item~\ref{dfitem:law-imposed game 1} of Definition~\ref{df:law-imposed game}. Then, $d'=d\in\Delta_a\cap L$ by the subcase assumption and statement~\eqref{eq:7-11-6}.
Thus, item~\ref{ppitem:reduce minimal gap-free to multiple minimal vertex cover 3} of the theorem is false by statement~\eqref{eq:7-11-7} and the case assumption $d'$ is safable.
\end{proof}

\subsection{Proof of Theorem~\ref{th:reduce gap-free reduction to vertex cover in reduced graph}}\label{app_sec:proof of Theorem reduce gap-free reduction to vertex cover in reduced graph}

To improve the readability, we first show several lemmas that will support the proof of Theorem~\ref{th:reduce gap-free reduction to vertex cover in reduced graph}.

\begin{lemma}\label{lm:reduced vertex cover in reduced graph}
For a graph $(V,E)$ and two sets $C', C\subseteq V$, the next two statements are equivalent:
\begin{enumerate}
    \item $C'\subseteq C$ and set $C'$ is a vertex cover of the graph $(V,E)$.
    \item $C$ is a vertex cover of the graph $(V,E)$ and $C'$ is a vertex cover of the subgraph $(C,E^C)$.
\end{enumerate}
\end{lemma}
\begin{proof}
($1\Rightarrow 2$)
Note that set $C$ is a vertex cover of the graph $(V,E)$ by the assumption $C\subseteq V$ of the lemma, item~1 above, and Lemma~\ref{lm:vertex cover monotonicity}.

Moreover, $C'\cap e\neq\varnothing$ for each edge $e\in E$ by item~\ref{dfitem:vertex cover cover} of Definition~\ref{df:vertex cover} and item~1 above.
Then, $C'\cap(C\cap e)=C'\cap C\cap e=C'\cap e\neq\varnothing$ for each egde $e\in E$ by the assumption $C'\subseteq C$ in item~1 above.
Hence, set $C'$ is a vertex cover of the graph $(C,E^C)$ by Definition~\ref{df:vertex set reduced graph}, item~\ref{dfitem:vertex cover cover} of Definition~\ref{df:vertex cover}, and the assumption $C'\subseteq C$ in item~1 above.

($1\Leftarrow 2$)
Note that, by item~2 above, item~\ref{dfitem:vertex cover cover} of Definition~\ref{df:vertex cover}, and Definition~\ref{df:vertex set reduced graph},
\begin{equation}\label{eq:7-11-8}
    C'\subseteq C\subseteq V
\end{equation}
and $C'\cap(C\cap e)\neq\varnothing$ for each egde $e\in E$.
Thus, $C'\cap e\supseteq C'\cap C\cap e\neq\varnothing$ for each egde $e\in E$.
Hence, $C'$ is a vertex cover of the graph $(V,E)$ by item~\ref{dfitem:vertex cover cover} and Definition~\ref{df:vertex cover} and statement~\eqref{eq:7-11-8}.
\end{proof}

\begin{lemma}\label{lm:reduce gap-free reduction to vertex cover in reduced graph 1}
For any two laws $L', L$ in a game $(\mathcal{A},\Delta,\mathbb{P})$, the next two statements are equivalent:
\begin{enumerate}
    \item $L'\subseteq L$ and law $L'$ is useful in the game $(\mathcal{A},\Delta,\mathbb{P})$.
    \item $L$ is a vertex cover in the graph $(\bigcup\Delta,\mathcal{S}(\mathbb{P}))$ and $L'$ is a vertex cover in the subgraph $(L,\mathcal{S}(\mathbb{P})^L)$.
\end{enumerate}
\end{lemma}
\begin{proof}
The statement of the lemma follows from Definition~\ref{df:law in game}, Theorem~\ref{th:useful=vertex cover}, and Lemma~\ref{lm:reduced vertex cover in reduced graph}.
\end{proof}

\begin{lemma}\label{lm:reduce gap-free reduction to vertex cover in reduced graph 2}
For any two laws $L', L$ in a game $(\mathcal{A},\Delta,\mathbb{P})$, an agent $a\in\mathcal{A}$, and a safable action $d\in\Delta_a^L$, the next two statements are equivalent:
\begin{enumerate}
    \item $L'\subseteq L$ and $d$ is a safe action of agent $a$ in the law-imposed game $(\mathcal{A},\Delta^{L'},\mathbb{P}^{L'})$.
    \item $L$ is a vertex cover in the graph $\mathcal{H}^{a,d}_{(\mathcal{A},\Delta,\mathbb{P})}$ and $L'$ is a vertex cover in the subgraph $(L,(\mathcal{E}^{a,d}_{(\mathcal{A},\Delta,\mathbb{P})})^L)$.
\end{enumerate}
\end{lemma}
\begin{proof}
Note that, by the assumption $L', L$ are two laws of the game $(\mathcal{A},\Delta,\mathbb{P})$ and Definition~\ref{df:law in game},
\begin{equation}\textstyle\label{eq:7-11-9}
    L',L\subseteq\bigcup\Delta.
\end{equation}

($1\Rightarrow 2$)
By Lemma~\ref{lm:safe action=vertex cover in reduced graph} and item~1 above, $L'$ is a vertex cover in the graph $\mathcal{H}^{a,d}_{(\mathcal{A},\Delta,\mathbb{P})}$.
Then, item~2 above is true by Lemma~\ref{lm:reduced vertex cover in reduced graph}, statement~\eqref{eq:7-11-9}, and the assumption $L'\subseteq L$ in item~1 above.

($1\Leftarrow 2$)
By Lemma~\ref{lm:reduced vertex cover in reduced graph} and item~2 above, $L'\subseteq L$ and
\begin{equation}\label{eq:7-10-3}
    \text{$L'$ is a vertex cover in the graph $\mathcal{H}^{a,d}_{(\mathcal{A},\Delta,\mathbb{P})}$.}
\end{equation}
Then, $d\in\Delta_a^L\subseteq\Delta_a^{L'}$ by item~\ref{dfitem:law-imposed game 1} of Definition~\ref{df:law-imposed game} and the assumption $d\in\Delta_a^L$ of the lemma.
Thus, $d$ is a safe action of agent $a$ in the law-imposed game $(\mathcal{A},\Delta^{L'},\mathbb{P}^{L'})$ by Lemma~\ref{lm:safe action=vertex cover in reduced graph}, statement~\eqref{eq:7-10-3}, and the assumption of the lemma that action $d$ is safable.
\end{proof}

\begin{lemma}\label{lm:reduce gap-free reduction to vertex cover in reduced graph 3}
For any two laws $L',L$ in a game $(\mathcal{A},\Delta,\mathbb{P})$, an agent $a\in\mathcal{A}$, and a safable action $d\in\Delta_a\cap L$, the next two statements are equivalent:
\begin{enumerate}
    \item $L'\subseteq L$ and $d$ is a safe action of agent $a$ in the law-imposed game $(\mathcal{A},\Delta^{L'},\mathbb{P}^{L'})$.
    \item $L\setminus\{d\}$ is a vertex cover in the graph $\mathcal{H}^{a,d}_{(\mathcal{A},\Delta,\mathbb{P})}$ and $L'$ is a vertex cover in the subgraph $(L\setminus\{d\},(\mathcal{E}^{a,d}_{(\mathcal{A},\Delta,\mathbb{P})})^{L\setminus\{d\}})$.
\end{enumerate}
\end{lemma}
\begin{proof}
Note that, by the assumption that $L', L$ are two laws of the game $(\mathcal{A},\Delta,\mathbb{P})$ and Definition~\ref{df:law in game},
\begin{equation}\textstyle\label{eq:7-11-10}
    L',L\subseteq\bigcup\Delta.
\end{equation}

($1\Rightarrow 2$)
Note that $d\in\Delta_a^{L'}$ by the second part of item~1 above and Lemma~\ref{lm:safe action in law imposed game}. Then, $d\notin L'$ by item~\ref{dfitem:law-imposed game 1} of Definition~\ref{df:law-imposed game}.
Thus, $d\in L\setminus L'$ by the assumption $d\in\Delta_a\cap L$ of the lemma.
Hence, by the first part of item~1 above,
\begin{equation}\label{eq:7-30-3}
    L'\subseteq L\setminus\{d\}.
\end{equation}
Moreover, $L\setminus\{d\}\subseteq (\bigcup\Delta)\setminus\{d\}$ by statement~\eqref{eq:7-11-10}.
Then, by statement~\eqref{eq:7-30-3} and item~\ref{dfitem:safe action verify graph vertex} of Definition~\ref{df:safe action verify graph},
\begin{equation}\textstyle\label{eq:7-11-11}
    L'\subseteq L\setminus\{d\}\subseteq \mathcal{V}^{a,d}_{(\mathcal{A},\Delta,\mathbb{P})}.
\end{equation}
Besides, by Lemma~\ref{lm:safe action=vertex cover in reduced graph} and the second part of item~1 above,
\begin{equation}\label{eq:7-18-2}
    \text{$L'$ is a vertex cover in the graph $\mathcal{H}^{a,d}_{(\mathcal{A},\Delta,\mathbb{P})}$}.
\end{equation}
Then, item~2 above follows from statements~\eqref{eq:7-11-11} and \eqref{eq:7-18-2} by Lemma~\ref{lm:reduced vertex cover in reduced graph} and Definition~\ref{df:safe action verify graph}.

($1\Leftarrow 2$)
By item~\ref{dfitem:vertex cover cover} of Definition~\ref{df:vertex cover} and item~2 above,
\begin{equation}\textstyle\label{eq:7-18-3}
    L'\subseteq L\setminus\{d\}\subseteq \mathcal{V}^{a,d}_{(\mathcal{A},\Delta,\mathbb{P})}.
\end{equation}
Then, $L'\subseteq L$ and $d\notin L'$.
Thus, by item~\ref{dfitem:law-imposed game 1} of Definition~\ref{df:law-imposed game} and the assumption $d\in\Delta_a\cap L$ of the lemma,
\begin{equation}\label{eq:7-30-4}
    d\in\Delta_a\setminus L'=\Delta_a^{L'}.
\end{equation}
Moreover, by Lemma~\ref{lm:reduced vertex cover in reduced graph} and statement~\eqref{eq:7-18-3}, item~2 above implies that $L'$ is a vertex cover in the graph $\mathcal{H}^{a,d}_{(\mathcal{A},\Delta,\mathbb{P})}$.
This, together with statement~\eqref{eq:7-30-4} and the assumption of the lemma that action $d$ is safable, further implies that $d$ is a safe action of agent $a$ in the law-imposed game $(\mathcal{A},\Delta^{L'},\mathbb{P}^{L'})$ by Lemma~\ref{lm:safe action=vertex cover in reduced graph}.
\end{proof}

\noindent\textbf{Theorem~\ref{th:reduce gap-free reduction to vertex cover in reduced graph}. }
\textit{For a gap-free law $L$ in a game $(\mathcal{A},\Delta,\mathbb{P})$, a law $L'$ is a gap-free reduction of $L$ if and only if at least one of the following statements is true:
\begin{enumerate}
    \item $L$ is a vertex cover of the graph $(\bigcup\Delta,\mathcal{S}(\mathbb{P}))$ and $L'$ is a vertex cover in the subgraph $(L,\mathcal{S}(\mathbb{P})^L)$;\label{ppitem:reduce gap-free reduction to vertex cover in reduced graph 1}
    \item there is an agent $a\in\mathcal{A}$ and a safable action $d\in\Delta_a^L$ such that $L$ is a vertex cover in the graph $\mathcal{H}^{a,d}_{(\mathcal{A},\Delta,\mathbb{P})}$ and $L'$ is a vertex cover in the subgraph $(L,(\mathcal{E}^{a,d}_{(\mathcal{A},\Delta,\mathbb{P})})^L)$;\label{ppitem:reduce gap-free reduction to vertex cover in reduced graph 2}
    \item there is an $a\!\in\!\mathcal{A}$ and a safable action $d\!\in\!\Delta_a\!\cap\! L$ such that $L\!\setminus\!\{d\}$ is a vertex cover in graph $\mathcal{H}^{a,d}_{(\mathcal{A},\Delta,\mathbb{P})}$ and $L'$ is a vertex cover in the subgraph $(L\!\setminus\!\{d\},(\mathcal{E}^{a,d}_{(\mathcal{A},\Delta,\mathbb{P})})^{L\setminus\{d\}})$.\label{ppitem:reduce gap-free reduction to vertex cover in reduced graph 3}
\end{enumerate}}
\begin{proof}
($\Rightarrow$)
Suppose $L'$ is a gap-free reduction of $L$.
Then, 
\begin{equation}\textstyle\label{eq:7-11-12}
    L'\subseteq L\subseteq\bigcup\Delta
\end{equation}
and $L'$ is gap-free in the game $(\mathcal{A},\Delta,\mathbb{P})$ by Definition~\ref{df:law in game} and item~\ref{dfitem:gap free reduction} of Definition~\ref{df:minimum gap free reduction}.
Thus, by Lemma~\ref{lm:gap free law}, one of the next cases must be true:

Case 1: $L'$ is useful in the game $(\mathcal{A},\Delta,\mathbb{P})$.
Then, item~\ref{ppitem:reduce gap-free reduction to vertex cover in reduced graph 1} of the theorem is true by statement~\eqref{eq:7-11-12} and Lemma~\ref{lm:reduce gap-free reduction to vertex cover in reduced graph 1}.

Case 2: there is an agent $a\in\mathcal{A}$ and a safe action $d$ of agent $a$ in the law-imposed game $(\mathcal{A},\Delta^{L'},\mathbb{P}^{L'})$.
Then, 
\begin{equation}\label{eq:7-11-13}
    d\in\Delta_a^{L'}
\end{equation}
by Lemma~\ref{lm:safe action in law imposed game}, and 
\begin{equation}\label{eq:7-18-1}
    \text{action $d$ is safable}
\end{equation}
by Definition~\ref{df:safable action}.
Note that $\Delta_a^{L'}\supseteq\Delta_a^{L}$ by statement~\eqref{eq:7-11-12} and item~\ref{dfitem:law-imposed game 1} of Definition~\ref{df:law-imposed game}.
Then, by statement~\eqref{eq:7-11-13}, one of the next two subcases must be true:

Subcase 2.1: $d\in\Delta_a^{L}$.
Then, by Lemma~\ref{lm:reduce gap-free reduction to vertex cover in reduced graph 2}, item~\ref{ppitem:reduce gap-free reduction to vertex cover in reduced graph 2} of the theorem follows from statements~\eqref{eq:7-11-12}, \eqref{eq:7-18-1}, and the Case~2 assumption.

Subcase 2.2: $d\in\Delta_a^{L'}\setminus\Delta_a^{L}$.
Note that $$\Delta_a^{L'}\setminus\Delta_a^{L}=\Delta_a\cap(L\setminus L')\subseteq\Delta_a\cap L.$$
Then, $d\in\Delta_a\cap L$ by the subcase assumption.
Thus, by Lemma~\ref{lm:reduce gap-free reduction to vertex cover in reduced graph 3}, item~\ref{ppitem:reduce gap-free reduction to vertex cover in reduced graph 3} of the theorem follows from statement~\eqref{eq:7-11-12}, \eqref{eq:7-18-1}, and the Case~2 assumption.

($\Leftarrow$)
Note that, by item~\ref{dfitem:gap free reduction} of Definition~\ref{df:minimum gap free reduction}, it suffices to show that, if any of the three items of the theorem is true, then $L'\subseteq L$ and $L'$ is a gap-free law in the game $(\mathcal{A},\Delta,\mathbb{P})$.

If item~\ref{ppitem:reduce gap-free reduction to vertex cover in reduced graph 1} is true, then $L'\subseteq L$ and $L'$ is useful in the game $(\mathcal{A},\Delta,\mathbb{P})$ by Lemma~\ref{lm:reduce gap-free reduction to vertex cover in reduced graph 1} and the assumption that $L,L'$ are both laws in the game.
Thus, $L'$ is a gap-free law in the game $(\mathcal{A},\Delta,\mathbb{P})$ by Lemma~\ref{lm:gap free law}.

If item~\ref{ppitem:reduce gap-free reduction to vertex cover in reduced graph 2} is true, then $L'\subseteq L$ and $d$ is a safe action of agent $a$ in the game $(\mathcal{A},\Delta^{L'},\mathbb{P}^{L'})$ by Lemma~\ref{lm:reduce gap-free reduction to vertex cover in reduced graph 2} and the assumption that $L,L'$ are both laws in the game.
Thus, $L'$ is a gap-free law in the game $(\mathcal{A},\Delta,\mathbb{P})$ by Lemma~\ref{lm:gap free law}.

If item~\ref{ppitem:reduce gap-free reduction to vertex cover in reduced graph 3} is true, then $L'\subseteq L$ and $d$ is a safe action of agent $a$ in the game $(\mathcal{A},\Delta^{L'},\mathbb{P}^{L'})$ by Lemma~\ref{lm:reduce gap-free reduction to vertex cover in reduced graph 3} and the assumption that $L,L'$ are both laws in the game.
Thus, $L'$ is a gap-free law in the game $(\mathcal{A},\Delta,\mathbb{P})$ by Lemma~\ref{lm:gap free law}.
\end{proof}

\subsection{Approximation Factor of $\mathtt{AppMinGFR}$}\label{app_sec:MinGFR approximation factor}

First, observe that all graphs in the statement of Theorem~\ref{th:reduce gap-free reduction to vertex cover in reduced graph} are $|\mathcal{A}|$-graphs by Lemma~\ref{lm:game to graph}, Definition~\ref{df:vertex set reduced graph}, and Definition~\ref{df:safe action verify graph}.
This means, all the $\mathtt{AppMinVC}$ input graphs in lines~35, 40, and 43 of Algorithm~\ref{alg:GFL design} are $|\mathcal{A}|$-graphs.

Moreover, by Theorem~\ref{th:reduce gap-free reduction to vertex cover in reduced graph}, the minimum gap-free reduction (say $R$) must be a minimum vertex cover in at least one of the $\mathtt{AppMinVC}$ input graphs (say the minimum vertex cover $C$ of graph $G$) in lines~35, 40, and 43 of Algorithm~\ref{alg:GFL design}:
\begin{equation}\label{eq:7-29-1}
    |R|=|C|.
\end{equation}
Then, by the assumption about the approximation factor of $\mathtt{AppMinVC}$,  the execution of $\mathtt{AppMinVC}$ on graph $G$ returns a vertex cover (say $C'$) whose size is at most $|\mathcal{A}|$ times the size of the minimum vertex cover $C$ in graph $G$:
\begin{equation}\label{eq:7-29-2}
    |C'|\leq |\mathcal{A}|\times |C|.
\end{equation}
Moreover, by lines~33, 35, 40, and 43 of Algorithm~\ref{alg:GFL design}, the size of the final output in line~44 of Algorithm~\ref{alg:GFL design} is no bigger than $|C'|$:
\begin{equation}\label{eq:7-29-3}
    |\mathit{output}|\leq |C'|.
\end{equation}
Hence, by statements~\eqref{eq:7-29-1}, \eqref{eq:7-29-2}, and \eqref{eq:7-29-3},
\begin{equation}\notag
    |output|\leq |\mathcal{A}|\times |R|,
\end{equation}
which means the size of the final output of $\mathtt{AppMinGFR}$ in Algorithm~\ref{alg:GFL design} is at most $|\mathcal{A}|$ times the size of the minimum gap-free reduction.
Therefore, $\mathtt{AppMinGFR}$ in Algorithm~\ref{alg:GFL design} is an $|\mathcal{A}|$-approximation of the problem $\mathbf{MinGFR}$.

\subsection{Algrithms for the $\mathbf{GFL}$ Problems}\label{app_sec:algorithms for GFL problems}

(See the next page.)

\SetKwComment{Comment}{//}{}
\begin{algorithm*}[bth]
\caption{Algorithms for the $\mathbf{GFL}$ Problems}
\label{alg:GFL design}
\SetKwFunction{FMain}{$\mathtt{IsSafable}$}
\SetKwProg{Fn}{}{:}{}
\Fn{\FMain{game $(\mathcal{A},\Delta,\mathbb{P})$, action $d$}}{
    \If{$d\notin\bigcup\Delta$}{
        \Return \textit{false}\;
    }
    \For{each profile $\delta\in\mathbb{P}$}{
        \If{$\mathcal{S}(\delta)=\{d\}$}{
            \Return \textit{false}\;
        }
    }
    \Return \textit{true}\;
}

\SetKwFunction{FMain}{$\mathtt{IsGFL}$}
\SetKwProg{Fn}{}{:}{}
\Fn{\FMain{game $(\mathcal{A},\Delta,\mathbb{P})$, set $L$}}{
    \If({\small\Comment*[f]{item~\ref{thitem:reduce gap-free to multiple vertex cover 1} of Theorem~\ref{th:reduce gap-free to multiple vertex cover}}}\normalsize){$\mathtt{IsVC}$ ($(\bigcup\Delta,\mathcal{S}(\mathbb{P}))$, $L$)}{
        \Return \textit{true}\;
    }
    \For({\small\Comment*[f]{item~\ref{thitem:reduce gap-free to multiple vertex cover 2} of Theorem~\ref{th:reduce gap-free to multiple vertex cover}}}\normalsize){each agent $a\in\mathcal{A}$ and each action $d\in\Delta_a$}{
        \If{$d\notin L$ and $\mathtt{IsSafable}$ ($(\mathcal{A},\Delta,\mathbb{P})$,$d$)}{
            \If{$\mathtt{IsVC}$ (graph $\mathcal{H}^{a,d}_{(\mathcal{A},\Delta,\mathbb{P})}$, set $L$)}{
                \Return \textit{true}\;
            }
        }
    }
    \Return \textit{false} \;
}

\SetKwFunction{FMain}{$\mathtt{IsMiniGFL}$}
\SetKwProg{Fn}{}{:}{}
\Fn{\FMain{game $(\mathcal{A},\Delta,\mathbb{P})$, set $L$}}{
    \If({\small\Comment*[f]{Check if set $L$ is a gap-free law}\normalsize}){not $\mathtt{IsGFL}$ ($(\mathcal{A},\Delta,\mathbb{P})$, $L$)}{
        \Return \textit{false}\;
    }
    \If({\small\Comment*[f]{item~\ref{thitem:reduce minimal gap-free to multiple minimal vertex cover 1} of Theorem~\ref{th:reduce minimal gap-free to multiple minimal vertex cover}}\normalsize}){$\mathtt{IsVC}$ ($(\bigcup\Delta,\mathcal{S}(\mathbb{P}))$, $L$) and not $\mathtt{IsMiniVC}$ ($(\bigcup\Delta,\mathcal{S}(\mathbb{P}))$, $L$)}{
        \Return \textit{false}\;
    }
    \For{each agent $a\in\mathcal{A}$ and each action $d\in\Delta_a$}{
        \If{$\mathtt{IsSafable}$ ($(\mathcal{A},\Delta,\mathbb{P})$,$d$)}{
            \eIf({\small\Comment*[f]{item~\ref{thitem:reduce minimal gap-free to multiple minimal vertex cover 2} of Theorem~\ref{th:reduce minimal gap-free to multiple minimal vertex cover}}\normalsize}){$d\notin L$}{
                \If{$\mathtt{IsVC}$ ($\mathcal{H}^{a,d}_{(\mathcal{A},\Delta,\mathbb{P})}$, $L$) and not $\mathtt{IsMiniVC}$ ($\mathcal{H}^{a,d}_{(\mathcal{A},\Delta,\mathbb{P})}$, $L$)}{
                    \Return \textit{false}\;
                }
            }({\small\Comment*[f]{item~\ref{thitem:reduce minimal gap-free to multiple minimal vertex cover 3} of Theorem~\ref{th:reduce minimal gap-free to multiple minimal vertex cover}}\normalsize}){
                \If{$\mathtt{IsVC}$ ($\mathcal{H}^{a,d}_{(\mathcal{A},\Delta,\mathbb{P})}$, $L\setminus\{d\}$)}{
                    \Return \textit{false}\;
                }
            }
            
        }
    }
    \Return \textit{true} \;
}

\SetKwFunction{FMain}{$\mathtt{AppMinGFR}$}
\SetKwProg{Fn}{}{:}{}
\Fn{\FMain{game $(\mathcal{A},\Delta,\mathbb{P})$, set $L$}}{
    \If({\small\Comment*[f]{Check if set $L$ is a gap-free law}\normalsize}){not $\mathtt{IsGFL}$ ($(\mathcal{A},\Delta,\mathbb{P})$, $L$)}{
        \Return \textit{false}\;
    }
    $\mathit{output}\leftarrow L$\;
    \If({\small\Comment*[f]{item~\ref{thitem:reduce gap-free reduction to vertex cover in reduced graph 1} of Theorem~\ref{th:reduce gap-free reduction to vertex cover in reduced graph}}\normalsize}){$\mathtt{IsVC}$ ($(\bigcup\Delta,\mathcal{S}(\mathbb{P}))$, $L$)}{
        $\mathit{output}\leftarrow \min\{\mathit{output},\mathtt{AppMinVC} (L,\mathcal{S}(\mathbb{P})^L)\}$\;
    }
    \For{each agent $a\in\mathcal{A}$ and each action $d\in\Delta_a$}{
        \If{$\mathtt{IsSafable}$ ($(\mathcal{A},\Delta,\mathbb{P})$,$d$)}{
            \eIf({\small\Comment*[f]{item~\ref{thitem:reduce gap-free reduction to vertex cover in reduced graph 2} of Theorem~\ref{th:reduce gap-free reduction to vertex cover in reduced graph}}\normalsize}){$d\notin L$}{
                \If{$\mathtt{IsVC}$ ($\mathcal{H}^{a,d}_{(\mathcal{A},\Delta,\mathbb{P})}$, $L$)}{
                    $\mathit{output}\leftarrow \min\{\mathit{output},\mathtt{AppMinVC} (L,(\mathcal{E}^{a,d}_{(\mathcal{A},\Delta,\mathbb{P})})^L)\}$\;
                }
            }({\small\Comment*[f]{item~\ref{thitem:reduce gap-free reduction to vertex cover in reduced graph 3} of Theorem~\ref{th:reduce gap-free reduction to vertex cover in reduced graph}}\normalsize}){
                \If{$\mathtt{IsVC}$ ($\mathcal{H}^{a,d}_{(\mathcal{A},\Delta,\mathbb{P})}$, $L\setminus\{d\}$)}{
                    $\mathit{output}\leftarrow \min\{\mathit{output},\mathtt{AppMinVC} (L\setminus\{d\},(\mathcal{E}^{a,d}_{(\mathcal{A},\Delta,\mathbb{P})})^{L\setminus\{d\}})\}$\;
                }
            }
        }
    }
    \Return $\mathit{output}$ \;
}
\end{algorithm*}

\end{document}